\def\C{\mathbf{ C}}
\def\N{\mathbf{ N}}
\def\Q{\mathbf{ Q}}
\def\R{\mathbf{ R}}
\def\Z{\mathbf{Z}}

\def\r#1#2{``#1''$\rightarrow$``#2''}
\documentclass{article}
\usepackage[hyphens]{url}
\usepackage{amsmath, amssymb,amsthm}
\usepackage{pdfpages}
\usepackage{verbatim}
\usepackage{hyperref}
\usepackage{combelow}
\usepackage[show]{ed}
\usepackage{graphicx}
\usepackage{xspace}
\usepackage{physics}
\usepackage[ruled,vlined]{algorithm2e}

\newtheorem{proposition}{Proposition}
\newtheorem{theorem}{Theorem}

\newtheorem{definition}{Definition}

\newtheorem{example}{Example}

\newtheorem{remark}{Remark}

\begin{document}

\def\r{$\rightarrow$}\noindent

\author{James H. Davenport$^{1}$, Matthew England$^{2}$, Scott McCallum$^3$
        }
\title{Enhanced CAD-Based Quantifier Elimination\\
       With Multiple Equational Constraints --\\
       \textbf{\textit{Preliminary Draft}}}  
\date{\small $^1$University of Bath, UK; $^2$Coventry University, UK; $^3$Macquarie University, Australia}

\maketitle

\begin{abstract}
This paper presents two enhancements to cylindrical algebraic decomposition (CAD) based quantifier elimination (QE) for cases in which multiple equational constraints are present in the given input formula $\phi^*$. The first enhancement provides more detail in the output when there is a conceptual partition of the set of variables of $\phi^*$ into parameters and unknowns. In such cases, we describe how to partition the parameter space so that: (1) in each open set of the partition the number $\nu$ of associated unknowns is a finite constant or is infinite; and (2) for each such open set for which $\nu$ is finite, an expression for the unknowns in terms of the parameters is provided. The second enhancement is an efficiency gain achievable in certain situations. Indeed, when certain conditions are met, the second CAD equational projection step can be reduced more significantly than is supported by the prior existing theory. Relevant theorems and worked examples for both enhancements are provided. Application areas include approximation theory, cuspidal manipulator classification, and biological/chemical systems.
\end{abstract}

\section{Introduction}
Quantifier elimination (QE) by cylindrical algebraic decomposition (CAD) is a general purpose tool for solving algebraic and geometric problems. While its many improvements have led to some success with its use in practical applications, there is room for further development of the method. Improvements which better address the requirements of special types of problems and take advantage of any favourable features of such problem types to reduce the amount of computation needed may be particularly beneficial. In this paper we present two enhancements to CAD based QE for when multiple equational constraints (ECs) are present in the given input formula. 

The first enhancement consists of providing more detail in the output, when appropriate, as follows. In classical QE for Tarski algebra, when given a quantified formula $\phi^*$ in prenex form, the output is usually expected to be a quantifier-free formula $\phi'$, involving the free variables of $\phi^*$ only, equivalent to $\phi^*$. However, it is sometimes the case that the given $\phi^*$ is itself the result of a careful formulation (or reformulation) of a problem of practical interest \cite{Collins1996, Jirstrand1997}. Practical mathematical questions often seek the values of, or information about, desired {\em unknowns}, given a set of data constituting the {\em parameters} of the problem. Particularly useful is some kind of closed-form expression for each unknown in terms of the parameters. Thus, in the case that a QE input $\phi^*$ comes with a clear partition of its set of variables into parameters and unknowns, we aim to provide more than a simple quantifier-free formula $\phi'$ describing the truth set of $\phi^*$. Indeed, we would like to provide a partition $\mathcal{P}$ of the 
parameter space, such that in each element of $\mathcal{P}$ the number 
of solutions (that is, the number of associated vectors
of unknowns) is a finite constant or is infinite. Ideally then, for 
each element of $\mathcal{P}$ having a finite positive constant number 
of solutions, we would like to find an expression for the solutions in 
terms of the parameters. This could be regarded as extending the scope 
of the work reported in \cite{LR07}, in which the given parametric 
polynomial systems are assumed to be quantifier-free and to satisfy 
other restrictions. Problems of the kind we treat are found in 
application areas such as 
approximation theory \cite{Collins1996}, 
cuspidal manipulator classification \cite{LR07}, 
control systems design \cite{Jirstrand1997}, 
and biological and chemical systems \cite{SE2022}.

Our second enhancement to CAD-based QE with multiple equational constraints (ECs) consists of identifying efficiency gains achievable in specific situations. Our aim here is to exploit the special structure of certain problem classes of genuine practical interest to ``\emph{push back the doubly exponential wall}'' of CAD-based QE. The metaphor is chosen to reflect that while the fundamental complexity of the general QE problem means such a wall exists, optimisations like those we present move the wall back bringing many new problems into the scope of practical solutions. The tantalizing, yet elusive, goal of identifying interesting sub-problem classes for which a major improvement to the upper complexity bound for CAD-based QE could be demonstrated also motivates us.

\subsection{Notation}

We use the following notation throughout, unless explicitly varied.
\begin{description}

\item[$\phi^*$:] The given input prenex formula of Tarski algebra.

\item[$n$:] The total number of variables, free and bound, in $\phi^*$. The
variables we generally use are $x_1,\ldots,x_n$.  CAD requires a fixed ordering of the variables, $\prec$, which determines the order in which they are eliminated:  we assume an ordering consistent with the variable indexing ($x_1 \prec \dots \prec x_n)$ with $x_n$ the first variable to be eliminated. We may then say a polynomial is of level $k$ if the highest ordered variable it contains is $x_k$.

In small examples, we may use simple letters without subscripts.
For instance, when $n=3$, we may use $x,y,z$ with $z$ the first to be eliminated.

\item[$k$:] The number of variables in $\phi^*$
that we consider to be bound (by quantifiers).
We assume these variables are $x_{n-k+1},\ldots,x_n$, consistent with the assumption that $\phi^*$ is a prenex formula. 
We denote the number of free variables of $\phi^*$ by $s = n-k$.

\item[$d$:] The maximum degree (\emph{in any individual variable}) of the polynomials in $\phi^*$.

\item[${\bf f}:=(f_1,\ldots,f_{t})$]: The polynomials in $\phi^*$ which are used in equational constraints for $x_{1},\allowbreak\ldots,x_n$.  We assume that $t \ge 1$; in fact, since we are primarily interested in the case of ``multiple'' equational constraints, we will usually assume that $t \ge k$. Note that no $f_i$ should lie in $\Z[x_1, \ldots,x_s]$, i.e. every equational constraint should involve a bound variable.

\end{description}
In summary, we may express $\phi^*$ in the form
\[ 
(Q_{n-k+1} x_{n-k+1}) \dots (Q_n x_n) \phi(x_1, \ldots, x_n)
\]
where $\phi$ is $f_1 = 0 \wedge \cdots \wedge f_{t} = 0 \wedge \psi$, for some quantifier-free formula $\psi$. We will further assume that $(Q_j x_j) = (\exists x_j)$, for every $j$. While we are primarily interested in enhancing QE for such formulas of Tarski algebra, we may sometimes consider analogous QE for constructible sets over the complex field $\C$.
 
\subsection{Examples}

\begin{example}
Let $\phi$ be $(x-1)^2+(y-2)^2=0$. Then $\phi$ is equivalent to $x=1 \wedge y=2$ over the reals $\R$. So, if we put $\phi^* \equiv \phi$ (that is, we apply no quantifiers to $\phi$) as input, we may find output $\phi' \equiv x = 1 \wedge y = 2$. If we put $\phi^* \equiv (\exists y) \phi$ as input, we may find output $\phi' \equiv x = 1$.

We remark that over $\C$ the situation is different. For example, if we put $\phi^* \equiv (\exists y) \phi$ as input, then we may find output $\phi' \equiv 0=0$.
\end{example}

\begin{example}
We consider an interesting problem based on \cite[Example 3.11]{FGT01}. The original example presents a Gr\"obner basis under the variable ordering 
$x \succ y \succ z \succ a \succ b \succ c$ for the ideal $I$  of $K[c,b,a,z,y,x]$ generated by
\[
\{a x^2 + x + y, \, b x + z, \, c y + y - z\},
\]
where $K$ is an arbitrary field. The context is the study of systems of parametric polynomial equations, in which $c,b,a$ are the parameters and $z,y,x$ are the unknowns. Accordingly, let us put $K = \R$, and define
\[
\phi^* \equiv (\exists z)(\exists y)(\exists x) [a x^2 + x + y = 0 
                                           \wedge b x + z = 0
                                           \wedge c y  + y - z = 0].
\]
For now, we make some simple observations concerning this problem, which 
we will discuss in more detail in later sections. First, it is not 
difficult to see quickly that $\phi^*$ is equivalent to $0=0$; that is, 
the truth set of $\phi^*$ in real $c,b,a$-space is the whole space. The reason is that, for all $c,b,a \in \R$, $x = y = z = 0$ is a solution of the given system. However, it is reasonable to ask for a description of all the solutions, as the parameter triple $(c,b,a)$ varies in $\R^3$. For example, when we put $c = b = a = 0$, we easily see that $(x,y,z) = (0,0,0)$ is the unique solution. 
Alternatively, suppose we put $c = -1$ and $b = 0$,
and take an arbitrary real $a > 0$. 
Then our system becomes $a x^2 + x + y = 0 \wedge z = 0$, which has infinitely many real solutions. Indeed, the set of all real solution triples 
is $\{(x,y,0) \in \R^3~|~a x^2 + x + y = 0\}$; that is, the parabola
$y = -a x^2 - x$ lying in the plane $z = 0$.
\end{example}
Subsequent sections will detail an enhancement to CAD-based QE for such parametric systems. This will allow a generic solution description for this example. Comparison with ideal theoretic methods applicable primarily for the case $K = \C$, but with relevance to the real case also, is anticipated in the final version of this paper.

\begin{example}
For this example, $n$ does not denote the number of variables as normal.  Solotareff's approximation problem \cite{Collins1996} asks one to find the best approximation, in the sense of the uniform norm on the interval $[-1,1]$, to a real polynomial of degree $n$ by a real polynomial  of degree at most $n-2$. The best approximation is known to be unique. The given polynomial may be taken to be $x^n + r x^{n-1}$, where $r$ is an arbitrary non-negative real number. Theory
in \cite{Achieser1956} provides explicit, computational expressions for the coefficients of the best approximation when 
$0 \le r \le S_n = n(\tan(\pi/(2n)))^2$. For $r > S_n$
the question may be advantageously formulated as a QE problem using some further theory in \cite{Achieser1956}, as explained in \cite{Collins1996}. 

We assume $r > S_n$ henceforth. The case $n=2$ is easy to solve by hand, while the paper \cite{Collins1996} describes computer-assisted solutions to this problem for the cases $n = 3,~4$, and a partial solution in the case $n=5$. Here we consider the QE formulation of the problem in the case $n=3$ presented in that paper. That is, noting that $S_3 = 1$, we wish to find the best approximation $ax + b$ to $x^3 + rx^2$, where $r > 1$. Using theory we deduce $a = 1$, and we obtain the following quantified formula which implicitly expresses $b$ in terms of $r$:
\begin{align*}
\phi^* \equiv (\exists u)[r > 1 \wedge -1 < u \wedge u < 1 
 &\wedge 3u^2 + 2ru - 1 = 0 \\
 & \wedge u^3 + ru^2 - u + r - 2b = 0].
\end{align*}
With the help of the QE program used in \cite{Collins1996} the quantifier-free formula
\begin{align*}
\phi' &\equiv [27b^2 - 2r^3 b - 36 rb + r^4 + 11r^2 - 1 = 0 \\
&\qquad  \wedge 27b - r^3 - 18r < 0 \wedge r > 1]
\end{align*}
was found as a solution. Insightful examination of the program's trace yields a still more explicit expression of $b$ in terms of $r > 1$: $b$ is the lesser of the two real roots of the polynomial on the left-hand side of the first conjunct in $\phi'$.
\end{example}
The publication of \cite{Collins1996} preceded the rigorous, partial validation of the equational constraint projection reduction technique used in that paper. (See Section 2.1 for more details about this technique.) The theory in \cite{McCallum1999} fully justified such reduction in the case of three variables; hence the solution of Solotareff in the case $n = 3$ reviewed in this example is confirmed. 
Solotareff's problem in the case $n = 4$ will be examined in detail in Section 5, following presentation of progress on equational projection validation in Section 4.






\subsection{Outline of the paper}

We outline the structure of the paper, and our contributions.
Section 2 provides a summary of relevant background material: we will assume that the reader is familiar with the CAD algorithm, its theory, and its application to QE; but will offer a summary of work published to date on the use of equational constraints, when present, to reduce the size of projection sets and to simplify the lifting process. So far, only a {\em semi-restricted} equational projection operation has been validated in general \cite{McCallum2001, EBD2020}. 

In Section 3, we describe a new enhancement to traditional ``single-step projection'' CAD-based QE with many equational constraints.
We describe an algorithmic framework for finding generic solutions of such QE problems representing parametric systems in which the output is more informative than that supplied by existing CAD-based QE tools.
In Section 4 we report progress on our efforts to validate Collins' {\em original} fully reduced equational projection operation, in the case that multiple equational constraints are present which satisfy certain conditions. 
Section 5 contains detailed discussion of two examples, in which we apply the work reported in Sections 3 and 4. Section 6 concludes the present paper, looking ahead to planned future contributions.

\section{Background material}

A \emph{cylindrical algebraic decomposition} (CAD) is a decomposition of $\R^n$ into cells (connected subsets), each of which is semi-algebraic, that are arranged cylindrically:  the projection of any two cells is either identical or disjoint, i.e. the cells stack up in cylinders above an induced CAD of $\R^{n-1}$.  The original CAD algorithm produced a CAD relative to a set of input polynomials such that each polynomial had constant sign ($-/0/+$) within each cell (\emph{sign invariance}), with subsequent algorithms developed to take logical formulae as input and produce corresponding \emph{truth-invariant} CADs. The main idea in traditional CAD algorithms is to identify a set of projection polynomials and use these to build the decomposition: with each cell either a section (where a projection polynomial vanishes) or a sector (the space between two sections).

Descriptions of the CAD algorithm, its variants and associated theory
can be found in \cite{ACM1984, Collins1975, CollinsHong1991, McCallum1988, McCallum1998}. We will make use of terminology and results from \cite{McCallum1988, McCallum1998}. The definitions
of key technical terms from those works (such as submanifold, analytic delineability and order-invariance) we use are summarised in \cite{McCallum1999}.

\subsection{Using equational constraints to reduce projection}
An {\em equational constraint (EC)} is an equation logically implied by the quantifier-free matrix of the prenex input formula $\phi^*$).  The idea of using an EC to reduce the projection operations and lifting process 
in CAD-based QE originated with Collins \cite{Collins1996, Collins1998}.
The basic idea is that, if $A \subset \Z[x_1, \ldots, x_n]$ is the set of
polynomials in $\phi^*$ and $f(x_1, \ldots, x_n) = 0$ is an EC (with $f$
assumed to be squarefree for simplicity), then the projection of $A$
may consist of only the discriminant of $f$, (a sufficient subset of) the coefficients of $f$,
and the resultant of $f$ and $g$, for each polynomial $g \in A \setminus \{f\}$ (where each $g$ is assumed to be prime to $f$, for simplicity).
Adjustment of the basic idea to take account of the factorization of
$f$ and the other polynomials $g$ into irreducible polynomials can 
easily be made. While this description of the idea focuses on the {\em first} projection (that is, elimination of the last variables $x_n$),
Collins implied that a similar reduction could be made for
subsequent projection steps, provided that suitable ECs are available
either explicitly or implicitly.
However, while the key claim of these initial works on the topic is
true in the special case $n = 2$ (proved below in
Proposition \ref{Bivariate case}), the validity of the claim
for $n > 2$ was not addressed in those original works: in short, the improvement was at 
that stage intuitive and conjectural.
The conference papers \cite{McCallum1999, McCallum2001} endeavoured to
justify the key claim rigorously, but succeeded in
providing only partial justification for the original approach
for $n > 3$ (as well its complete verification for $n = 3$).
The reader is referred to the comprehensive work \cite{EBD2020}
for a detailed synthetic review of this subject.

We recall the main fundamental lemma from \cite{McCallum1999} for the partial validation of the reduced equational projection.

\begin{theorem}[Theorem 2.2 in \cite{McCallum1999}] \label{Main lemma for reduced equational projection}
Let $n \ge 2$, let $f(x_1, \ldots, x_n)$ and $g(x_1, \ldots, x_n)$
be real polynomials of positive degrees in the main variable $x_n$,
let $R(x_1, \ldots, x_{n-1})$ be the resultant of $f$ and $g$, and suppose
that $R \neq 0$. Let $S$ be a connected subset of $\R^{n-1}$ on
which $f$ is delineable and in which $R$ is order-invariant.
Then $g$ is sign-invariant in each section of $f$ over $S$.
\end{theorem}

The above theorem underpins the validity argument for Collins'
original, fully reduced, equational projection scheme for the
\emph{first} projection step (the different invariance properties of the hypothesis and conclusion block its repeated application for subsequent steps).
Theorem 2.3 of \cite{McCallum1999} formalizes this validation argument.
In general, subsequent projection steps
can use only a semi-restricted projection scheme, when
suitable equational constraints (explicit or implicit) are available
\cite{McCallum2001, EBD2020}. Section 4 below presents certain
extensions of the above theorem, as well as some other new results, 
which are useful in validating
the original, fully reduced equational projection for both the first
and second projection steps, under some special assumptions.

For the reader's convenience, we summarize below the definitions of
the fully- and semi-reduced equational projection sets
from \cite{McCallum1999, McCallum2001, EBD2020}.
Here, $A$ denotes an irreducible basis in $\Z[x_1, \ldots, x_n]$,
$E \subset A$, and $P(A)$ denotes the projection set defined
in \cite{Brown2001} (the ``Brown-McCallum projection''):
\begin{align*}
P_E(A) &= P(E) \cup \{\mathrm{res}_{x_n} (f,g)~|~ f \in E \wedge
                        g \in A \setminus E\}; \\
P_E^*(A) &= P_E(A) \cup \mathrm{discr}(A \setminus E)
                    \cup \mathrm{ldcf}(A \setminus E).
\end{align*}
\begin{remark}
    The sets $\mathrm{discr}(A \setminus E)$ and 
    $\mathrm{ldcf}(A \setminus E)$ represent sets of discriminants
    and leading coefficients, respectively, with respect to $x_n$ of elements of $A \setminus E$.
\end{remark}
\begin{remark}
    When $A \subset \Z[x_1, \ldots, x_n]$ is arbitrary, $P_E(A)$ is
    defined to be $\mathrm{cont}(A) \cup P_F(B)$, where 
    $\mathrm{cont}(A)$ denotes the set of non-zero contents with respect to $x_n$ of the elements of $A$, $B$ is the finest squarefree basis 
    for the set $\mathrm{prim}(A)$ of all primitive parts with respect to
    $x_n$ of the elements of $A$ which have positive degree in $x_n$,
    and $F$ is the finest squarefree basis for $\mathrm{prim}(E)$.
    For $A$ arbitrary, $P_E^*(A)$ is defined analogously.
\end{remark}
\begin{remark}
    The semi-reduced equational projection set $P_E^*(A)$ was first defined in \cite{McCallum2001}, 
    but a suitable coefficient set for $A \setminus E$
    was erroneously omitted; this was corrected in \cite{EBD2020}.
\end{remark}
\begin{remark}
    In the sources cited for the above definitions, $P(A)$ denoted
    the McCallum projection set. However, the use of the Brown-McCallum
    projection set instead is likely to be more efficient. It remains
    valid for CAD-based QE where the originally defined equational
    projection sets are valid, provided that the set of polynomials 
    in the input formula is well-oriented \cite{McCallum1998},
    the finitely many points over which projection factors
    vanish identically are computed,
    and such points are ``added'' during the lifting phase of CAD 
    construction \cite{Brown2001}.
\end{remark}

It was stated in \cite{McCallum1999} (and recalled above) 
in relation to Collins' original,
fully reduced, equational projection scheme that
``the validity of the scheme is certainly clear in the bivariate case''.
However, no general proof of this claim was provided in that work,
nor has it been supplied in any subsequent paper (as far as we know).
So, for completeness, we offer here the following slight strengthening
of Theorem 2.3 of \cite{McCallum1999} in the bivariate case.
For convenience, we use the notion of an {\em irreducible basis}
in a polynomial ring \cite{MPP2019}.

\begin{proposition} \label{Bivariate case}
Let $A$ be an irreducible basis in $\Z[x,y]$.
Let $E \subseteq A$ and let $S$ be an open interval in $\R^1$.
Suppose each element of $P_E(A)$ is order-invariant in $S$.
Then each element of $E$ is analytic delineable on $S$,
the sections over $S$ of the elements of $E$ are pairwise disjoint,
and each element of $A$ is order-invariant in every such section.
\end{proposition}

\begin{proof}
    Let $f_i(x,y) \in E$. By definition of $P_E(A)$, the leading
    coefficient $l_i(x)$ of $f_i$ belongs to $P_E(A)$.
    Since $S$ is an open interval of the real line
    and $l_i(x)$ is order-invariant in $S$, by hypotheses,
    it follows that $l_i(x)$ vanishes nowhere in $S$.
    Also, by hypothesis, the discriminant $d_i(x)$ of $f_i(x,y)$
    is order-invariant in $S$; and the resultant $r_{i,j}(x)$
    of $f_i$ and every other $f_j \in E$ is order-invariant in $S$.
    Let $f(x,y)$ denote the product of all the $f_i(x,y)$ in $E$.
    Then the leading coefficient $l(x)$ of $f(x,y)$,
    which is the product of all the $l_i(x)$,
    vanishes nowhere in $S$. Therefore $f(x,y)$ is degree-invariant
    in $S$, of positive degree. Also, the discriminant $d(x)$
    of $f(x,y)$, which is the product of all the $d_i(x)$ and
    the squares of all the $r_{i,j}(x)$ 
    (where all the $d_i(x)$ and all the $r_{i,j}(x)$ are elements of
    $P_E(A)$), is order-invariant in $S$. Hence, by Theorem 2 of
    \cite{McCallum1998},
    $f$ is analytic delineable on $S$ and is order-invariant
    in each of its sections over $S$.
    The first two conclusions of the proposition follow
    by Lemma A.7 of \cite{McCallum1988}.
    The order-invariance of each element of $E$ in every section
    of $f$ over $S$ follows by Lemma A.3 of \cite{McCallum1988}.
    Finally, consider an arbitrary element $g \in A \setminus E$,
    and again let $f_i \in E$.
    By definition of $P_E(A)$, we know $\mathrm{res}_y(f_i,g) \in P_E(A)$.
    Therefore, by the order-invariance hypothesis,
    this polynomial vanishes nowhere in $S$.
    It follows that $g$ vanishes at no point of a section of $f_i$
    over $S$; hence $g$ is order-invariant (of order 0)
    in every section of $f_i$ over $S$.
\end{proof}






\newpage

\section{Framework for enhanced CAD-based QE
            for parametric systems with many ECs}
Recall that we assume throughout this paper that our given input formula
$\phi^*$ has the form
\[(\exists x_{n-k+1}) \cdots (\exists x_n) \phi(x_1, \ldots, x_n),\]
where $\phi \equiv [f_1 = 0 \wedge \cdots \wedge f_t = 0] \wedge \psi$,
for some quantifier-free formula $\psi$ (see Section 1.1).  
Recall also that we denote the free variables of $\phi^*$ by $x_1, \ldots, x_s$,
where $s = n-k$, and let us suppose $t \ge k$.
For this subsection we shall further assume that an initial segment of
our list of $n$ variables is partitioned into $p$ {\em parameters} and
$u$ {\em unknowns} thus: $x_1, \ldots, x_p, x_{p+1}, \ldots, x_{p+u}$.
In general, the only restrictions on $p$ and $u$ are that $p \le s$ and
$p+u \le n$. However, for simplicity, we shall consider only the
case $p = s \ge 1$, $u \ge 1$ and $p+u = n$. Adjustment of our framework described below for other cases which may arise in practice would be straightforward.

We will present our proposed CAD-based method 
for finding generic solutions of a parametric
system given by a quantified formula $\phi^*$ as described above.
We first formalize our terminology relating to solutions.
\begin{definition}
Let $\alpha = (\alpha_1, \ldots, \alpha_s) \in \R^s$ be a specific
real parameter vector. We shall call any specific vector of real
unknowns $\beta = (\beta_{s+1}, \ldots, \beta_n) \in \R^k$
for which $\phi(\alpha, \beta)$ is true a
{\em solution vector for $\phi$ associated with} $\alpha$.
\end{definition}
\begin{definition}
A {\em system of generic solutions} for $\phi^*$
consists of a list $\mathcal{O}$ of pairwise disjoint
connected open subsets of parameter space $\R^s$ such that:
\begin{enumerate}
 \item the closure of $\cup \mathcal{O}$ is $\R^s$;
 \item for every subset $S \in \mathcal{O}$, we have that for any $\alpha \in S$, the number of distinct solution vectors for $\phi$ associated
    with $\alpha$ is a constant, $\nu_S$ of 
    $\N \cup \{\infty\}$, as $\alpha$ varies within $S$; and
  \item for each subset $S \in \mathcal{O}$ for which $\nu_S$ is finite
    and positive, the $\nu_S$ distinct solution vectors for $S$ 
    are expressible as continuous functions of the parameters in $S$.
\end{enumerate}
\end{definition}


Our method for finding a description of a system of generic
solutions for a given $\phi^*$ is described formally 
in Algorithms 1 and 2, the latter being a sub-algorithm of the former. 
Our method is essentially a synthesis of existing methods and concepts
(such as QE by partial CAD, equational constraint projection, etc.),
together with an extension to the stack construction phase of partial CAD.
A key notion for the method is that of a {\em pivot} constraint
\cite{Collins1998}:
during the projection phase, when more than one equational constraint
is available at a given level, one such must be chosen and designated
as the pivot constraint at that level.
The next projection set computed is then semi-restricted with respect to
the (set of irreducible factors of the) pivot, as reviewed in Section 2.1.
As a prelude to our description of Algorithm 1 we present a foundational
result that underpins it.

\begin{proposition} \label{Invariance of solution vector cardinality}
    Let $\mathcal{D}$ be a truth-invariant CAD of $\R^n$ 
    for $\phi$ and let $\mathcal{D}_s$ denote the CAD of $\R^s$ induced 
    by $\mathcal{D}$.
    Let $c$ be a cell of $\mathcal{D}_s$,
    and let $\alpha \in c$. Then the total number $\nu_{c, \alpha}$
    of distinct solution vectors $\beta \in \R^k$ for $\phi$
    over $\alpha$ is an element of 
    $\N \cup \{\infty\}$, constant as $\alpha$ varies in $c$:
    hence we may define $\nu_c$, the cardinality of the set of
    solution vectors associated with the points of $c$, to be this number.
    Moreover, in the case that $0 < \nu_c < \infty$, the $\nu_c$ distinct
    solution vectors for $c$ are expressible as continuous functions
    of the parameters in $c$.
\end{proposition}

\begin{proof}
    Recall that $1 \le k \le n-1$ and $s = n - k$ (from the preamble). 
    Define the proposition
    $\mathcal{P}(k)$ as follows: for every cell $c \in \mathcal{D}_{n-k}$
    and for every $\alpha \in c$, the number $\nu_{c, \alpha}$
    is a constant element of $\N \cup \{\infty\}$, as $\alpha$
    varies in $c$. We shall prove that $\mathcal{P}(k)$ is true for all
    $k$ in the range $1 \le k \le n-1$ by induction on $k$.

    The induction base is proved as follows.
    Let $c \in \mathcal{D}_{n-1}$. First suppose that the cylinder over $c$
    contains some sector-cell $\hat{c}$ of $\mathcal{D}$
    for which $\phi$ is true throughout $\hat{c}$.
    Take any $\alpha \in c$.
    Then, by the definition of sector-cell,
    there are infinitely many $\beta \in \R$ such that
    $(\alpha, \beta)$ lies in $\hat{c}$, hence for which
    $\phi(\alpha, \beta)$ is true.
    Therefore, $\nu_{c, \alpha} = \infty$.
    We have shown that $\nu_{c, \alpha}$ is constant (and infinite)
    for all $\alpha \in c$.
    Now suppose, on the other hand, that the cylinder over $c$
    contains no sector-cell for which $\phi$ is true.
    Take any $\alpha \in c$: then $\nu_{c, \alpha}$ equals the total number of section-cells
    of $\mathcal{D}$ over $c$ for which $\phi$ is true, a number clearly
    independent of $\alpha$. Again, we have shown that $\nu_{c, \alpha}$
    is constant for all $\alpha \in c$ (but in this case, it is finite).
    The induction base $\mathcal{P}(1)$ is now established.

    For the induction step, let $k > 1$ and assume $\mathcal{P}(k-1)$ is true
    (as induction hypothesis). We must show that $\mathcal{P}(k)$ holds.
    Let $c \in \mathcal{D}_{n-k}$.
    Suppose that the cylinder over $c$ contains some 
    cell $\hat{c} \in \mathcal{D}_{n-k+1}$ for which $\nu_{\hat{c}}$
    is infinite (noting that the notation $\nu_{\hat{c}}$
    is justified by the induction hypothesis).
    Then the number $\nu_{c, \alpha}$, for points $\alpha \in c$,
    is constant and infinite also.
    So, we may assume henceforth that the cylinder over $c$ contains
    no cell $\hat{c}$ in $\mathcal{D}_{n-k+1}$ for which $\nu_{\hat{c}}$
    is infinite. Two cases remain to be considered.
    Firstly, consider the case in which the cylinder over $c$
    contains some sector-cell $\hat{c} \in \mathcal{D}_{n-k+1}$
    for which $\nu_{\hat{c}}$ is positive
    (noting again that the notation is justified by the induction
    hypothesis). Then the number $\nu_{c, \alpha}$, for points
    $\alpha \in c$, is constant and infinite. The reason is that,
    with $\alpha \in c$ arbitrary,
    for each of the infinitely many $\beta_{s+1} \in \R$
    such that $(\alpha, \beta_{s+1}) \in \hat{c}$,
    there is a solution vector $(\beta_{s+2}, \ldots, \beta_n)$
    for $\phi$ associated with $(\alpha, \beta_{s+1})$.
    Secondly, consider the case in which the cylinder over $c$
    contains no such sector-cell. In this case, the number
    $\nu_{c, \alpha}$, for points $\alpha \in c$, equals the sum 
    over the sections $\hat{c}$ of $\mathcal{D}_{n-k+1}$
    of the finite quantities $\nu_{\hat{c}}$ (noting again the use of the
    induction hypothesis), a sum which is independent of $\alpha$.
    This completes the proof of the induction step,
    establishing the truth of $\mathcal{P}(k)$ for all $k$, with $1 \le k \le n-1$.

    Define the proposition $\mathcal{Q}(k)$ as follows:
    for every cell $c \in \mathcal{D}_{n-k}$,
    for which $0 < \nu_c < \infty$,
    the $\nu_c$ distinct solution vectors for $c$ are expressible
    as continuous functions of the parameters in $c$.
    Then $\mathcal{Q}(k)$ may be proved true for all $k$ in the range 
    $1 \le k \le n-1$ by induction on $k$,
    by analogy with (and suitable modification of)
    the proof for proposition $\mathcal{P}$ provided above.
    
\end{proof}

\begin{algorithm}
\caption{Generic Solutions for a Parametric System}

\medskip
\noindent 
    $(\mathcal{I}, \mathcal{S}, \mathcal{N}, \mathcal{F}) \leftarrow \mathrm{GSPS}(\phi^*)$

\medskip
\emph{Input}: $\phi^*$ is a formula of Tarski algebra of the form
            $\phi^* \equiv (\exists x_{n-k+1}) \cdots (\exists x_n)
            \phi(x_1, \ldots, x_n)$, where
        $\phi \equiv [f_1 = 0 \wedge \cdots \wedge f_t = 0] \wedge \psi$,
        for some quantifier-free formula $\psi$, and $t \ge k \ge 1$.
        With $s = n-k \ge 1$, $\phi^*$ is considered to be a parametric system in which $x_1, \ldots, x_s$ are parameters and 
        $x_{s+1}, \ldots, x_n$ unknowns.

\medskip
\emph{Output}: ${\mathcal I}$ and $\mathcal{S}$ are lists of indices
and sample points, respectively, 
for the $s$-cells of the CAD $\mathcal{D}_s$ of $\R^s$ induced by a
truth-invariant CAD $\mathcal{D}$ of $\R^n$ for $\phi$.
$\mathcal{N}$ is the list of numbers $\nu_c$, for each cell $c$ of
$\mathcal{D}_s$ (where $\nu_c$ is defined by the above proposition).
${\mathcal F}$ is a list $\{F_c~|~c \in {\mathcal D} \wedge
0 < \nu_c < \infty\}$ such that every $F_c$ is a list of
suitable functional expressions of the $\nu_c$ associated unknown vectors
in terms of the parameter vector $x \in c$.

\medskip
\begin{itemize}
    \item[1.] \textbf{Initialization.} Set $A \leftarrow ~$ the set of 
    polynomials which occur in $\phi$ and
    $E \leftarrow \{f_1, \ldots, f_t\}$. 
    Set $J_n \leftarrow A$ and $C_n \leftarrow$ the set of level $n$ polynomials in $E$.
    If $C_n \neq \emptyset$ then 
        set $e_n \leftarrow$ to an element of $C_n$ and
        $E_n \leftarrow \{e_n\}$ 
    else
        set $e_n \leftarrow 0$. 
    [$C_n$ is the candidate set of
    level $n$ EC polynomials, and $e_n$ is the level $n$ pivot
    if an EC exists.]

\medskip
     \item[2.] \textbf{Compute remaining pivots.}
     For $j \leftarrow n-1$ down to $n-k+1$: 
        set $C_j \leftarrow$ the candidate set of level $j$ EC polynomials;
        if $C_j \neq \emptyset$ then 
            set $e_j \leftarrow$ an element of $C_j$ and $E_j \leftarrow \{e_j\}$ 
            else set $e_j \leftarrow 0$.
     [$C_{n-1}$ is the union of the set of all level $n-1$ elements
     of $E$ and the set of level $n-1$ resultants with respect to $x_n$
     of pairs of distinct elements of $C_n$, for which one element
     of each pair is $e_n$. General definition of $C_j$ is needed.]

\medskip
     \item[3.] \textbf{Compute projection sets.}
     For $j \leftarrow n$ down to $n-k+1$: 
     if $e_j \neq 0$ then 
        set $J_{j-1} \leftarrow P_{E_j}^*(J_j)$;
     else
        set $J_{j-1} \leftarrow P(J_j)$.
     [Where possible, in the case that $e_j \neq 0$, 
     replace $P_{E_j}^*(J_j)$ by the fully
     restricted projection set $P_{E_j}(J_j)$. See Section 2.1
     for a review of the definitions of these projection sets.
     Note, in particular, that construction of each set $J_{j-1}$ entails
     computation of $\mathrm{cont}(J_j)$, $\mathrm{prim}(J_j)$, 
     and irreducible basis for the latter.]

\medskip
     \item[4.] \textbf{Analyse projection factor coefficient sets.}
     For each projection factor of level 3 upwards, examine zeros
     of the coefficient system. If a positive dimensional solution set is found then report $A$ as not well-oriented and exit. Otherwise solve the coefficient system and retain solutions for lifting \cite{Brown2001}.

\end{itemize}
\end{algorithm}

\setcounter{algocf}{0}

\begin{algorithm}[ht]
\caption{continued}
\begin{itemize}

     \item[5.] \textbf{Compute open cells of a $J_s$-invariant 
     CAD of parameter space $\R^s$
     together with solution vector cardinality information.}
     Compute lists $\mathcal{I}$ and $\mathcal{S}$ for the $s$-cells
     of a CAD ${\mathcal D}_s$ of $\R^s$ such that every element 
     of $J_s$ is order-invariant in each cell constructed.
     [See \cite{McCallum1993} or \cite{LR07} for an
     efficient algorithm for this task.] 
     For each cell $c$ constructed, with rational
     sample point $\alpha_c \in \mathcal{S}$, construct stacks which, 
     together with $c$, form a tree $\mathcal{T}_c$
     of cells rooted at $c$ sufficient to determine the cardinality
     $\nu_c$ of the set of associated solution vectors in $\R^k$.

\medskip

     \item[6.] \textbf{Construct ${\mathcal F}$.} For each $s$-cell $c$ of ${\mathcal D}_s$ with sample point $\alpha_c \in \Q^s$, and with
     $0 < \nu_c < \infty$, use sub-algorithm $\mathrm{CFEL}$ to construct
     a list $F_c$ of functional expressions of the $\nu_c$ associated solution vectors in $\R^k$ in terms of the parameters $x \in c$.
     Exit.
\end{itemize}
    
\end{algorithm}

\begin{algorithm}[!ht]
\caption{Construct Functional Expression List}

\medskip
\noindent
$F \leftarrow \mathrm{CFEL}(\alpha, \mathcal{B}, B)$

\medskip
\emph{Input}:
$\alpha$ is the sample point of a cell $c \in \R^s$,
which is the root of a tree $\mathcal{T}_c$ of cells whose non-root nodes
are partitioned into stacks.
[We say $c$ is at level $s$, its children are at level $s+1$, etc.]
$\mathcal{B}$ is a non-empty list of vectors $\beta \in \R^k$ such that
the set of sample points for the cells of $\mathcal{T}_c$ at level $n$ is
$\{(\alpha, \beta)~|~\beta \in \mathcal{B}\}$.
$B$ is a list $(B_{s+1}, \ldots, B_n)$, in which each $B_i$ is an irreducible basis at level $i$ such that 
$\prod_{p_{i} \in B_{i}} p_{i}$
is analytic delineable over every cell of $T$ at level $i-1$.

\medskip
\emph{Output}:
$F$ is a list of functional expressions for the cells of $\mathcal{T}_c$
at level $n$ over $c$. That is, for every cell $c_n$ of 
$\mathcal{T}_c$ at level $n$,
$F$ contains the description of an analytic function 
$\theta~:~c \rightarrow \R^k$ whose graph is $c_n$.

\medskip
\begin{itemize}
    \item[1.] \textbf{Initialization.} Set $F$ to be the empty list.

    \medskip
    \item[2.] \textbf{Construct every needed functional expression.}
    Remove a vector $\beta$ from the list $\mathcal{B}$.
    For $i \leftarrow s+1$ to $n$ perform the following actions.
    (Let $p_{i}(x, x_{s+1}, \ldots, x_i) \in B_i$ be such that
    $p_{i}(\alpha, \beta_{s+1}, \ldots, \beta_i) = 0$.
    Suppose $\beta_i$ is the $k_i$th real root of 
    $p_i(\alpha, \beta_{s+1}, \ldots, \beta_{i-1}, x_i)$ and
    that real valued functions 
    $\theta_{s+1}, \ldots, \theta_{i-1}$, each with domain $c$, have been defined. Define $\theta_i~:~c \rightarrow \R$ by letting
    $\theta_i(x)$ be the $k_i$th real root of
    $p_i(x, \theta_{s+1}(x), \ldots, \theta_{i-1}(x), x_i)$.)
    Set $\theta \leftarrow (\theta_{s+1}, \cdots, \theta_n)$.
    Append $\theta$ to the list $F$.

    \medskip
    \item[3.] \textbf{Is $\mathcal{B}$ non-empty?}
    If $\mathcal{B}$ is non-empty go back to Step 2; else Exit.
\end{itemize}
\end{algorithm}

Some remarks are in order. 
First, Algorithm 1 is based upon the
equational variant of algorithm $\mathrm{QEPCAD}$ 
(Quantifier Elimination by 
Partial CAD) which is sketched in \cite{McCallum2001}. 
It is also influenced
by \cite{LR07}: the specification of our 
$\mathrm{GSPS}$ is in accord with the framework and 
concepts of \cite{LR07};
and our use of ``open CAD'' in step 5 of $\mathrm{GSPS}$ recalls its
mention in \cite{LR07}.
As in \cite{LR07}, the focus is on generic solutions of a system,
since these are likely to be of greatest practical interest
and can be computed more efficiently.
Straightforward changes can be made to the above algorithms
to yield full solutions, but at greater computational cost.
The validity of Algorithm 1 follows from the correctness of
the equational variant of $\mathrm{QEPCAD}$ which is proved in 
\cite{McCallum2001},
the validity of the Brown-McCallum projection operation \cite{Brown2001},
the correctness of the ``open CAD'' algorithm \cite{McCallum1993, 
LR07}, and Proposition \ref{Invariance of solution vector cardinality}.

Second, it must be admitted that the more informative output specified
by $\mathrm{GSPS}$ may take much longer to compute than the
relatively simpler output provided by $\mathrm{QEPCAD}$.
The reason is that, for each cell $c \in \mathcal{D}$, 
$\mathrm{QEPCAD}$ 
can abort the stack construction process based at $c$ and its 
descendants as soon as the
truth value of $\phi^*$ on $c$ is determined. However, it may require
many more stack constructions based at $c$ and its descendants to 
compute the desired number $\nu_c$, rather than just determining its
sign. The construction of the list $\mathcal{F}$ is 
likely to require a similar amount of additional work.

Third, in view of the additional work likely to be needed for complex 
problem instances, its full application may be useful only in relatively 
simple or special situations. For example, problems with a small number 
of unknown variables, or those having several unknowns but for which 
theory predicts a small number of solutions for each parameter vector,
may be treatable by our algorithmic framework. Section 5 revisits 
Examples 1 and 2, demonstrating some concrete benefits of 
$\mathrm{GSPS}$ applied to these specific problems.

Fourth, even if one is faced with a complex problem, our algorithmic 
framework applied to the problem in a partial way may yield valuable 
insight into the solutions.

Fifth, we are hopeful that further progress on the theory of equational 
projection set reduction may yield more efficiencies for the method. The 
next section reports on some limited progress in that direction.

We conclude this section with some simple examples.

\begin{example}
    Let $n = 2$ and $\phi^* \equiv (\exists y)[y^2 - x = 0]$.
    Given this input, algorithm $\mathrm{GSPS}$ yields descriptions
    of the two 1-cells $c_1 = (-\infty, 0)$ and $c_2 = (0, \infty)$
    in $\mathcal{D}_1$. For $c_1$, $\nu = 0$; while for $c_2$,
    $\nu = 2$. Put $c = c_2$, for simplicity.
     Algorithm $\mathrm{CFEL}$ then produces a list $F_c = (\theta, \psi)$,
     where $\theta : c \rightarrow \R$ is given by
     $\theta(x) = -\sqrt{x}$ and $\psi : c \rightarrow \R$
     is given by $\psi(x) = \sqrt{x}$.
     (Strictly speaking, $\theta(x)$ is expressed as 
     ``the first (smaller) real root of $y^2 - x$'', etc.)
     Informally, for the first solution over $c$,
     the unknown $y$ is expressed in terms of the parameter $x$ by
     $y = \theta(x)$; and for the second solution over $c$,
     $y = \psi(x)$.
\end{example}

\begin{example}
    Let $n = 3$ and 
    $\phi^* \equiv 
    (\exists y)(\exists z)[z^2 + y^2 + x^2 - 1 = 0 \wedge z - y = 0]$.
    Given this input, algorithm $\mathrm{GSPS}$ yields descriptions
    of the three 1-cells:\\ $c_1 = (-\infty, -1)$,
    $c_2 = (-1, 1)$ and $c_3 = (1, \infty)$.
    For $c_1$ and $c_3$, $\nu = 0$; while for $c_2$, $\nu = 2$.
    Put $c = c_2$, for simplicity.
    Algorithm $\mathrm{CFEL}$ then produces a list
    $F_c = (\theta, \psi)$.
    We have $\theta = (\theta_2, \theta_3)$, where
    $\theta_2(x) = -\sqrt{(1 - x^2)/2}$ and
    $\theta_3(x) = \theta_2(x)$.
    We have $\psi = (\psi_2, \psi_3)$, where
    $\psi_2(x) = \sqrt{(1 - x^2)/2}$ and
    $\psi_3(x) = \psi_2(x)$.
    (Again, strictly speaking, $\theta_2(x)$ is expressed as
    ``the first (smaller) real root of $2y^2 + x^2 - 1$'', etc.)
    Informally, for the first solution over $c$, the unknowns $y$
    and $z$ are expressed in terms of the parameter $x$ by
    $y = \theta_2(x)$ and $z = \theta_3(x)$; 
    and for the second solution over $c$,
    $y = \psi_2(x)$ and $z = \psi_3(x)$.
\end{example}

\newpage

\section{Progress on validating use of fully reduced equational 
projection}

In this section
we present variations of existing theorems relating to delineability. 
These are of help in justifying the use of the fully reduced equational 
projection operation beyond the first projection step 
in certain special situations, consistent with the framework
described in the previous section.

In the first subsection we review some existing published work in this 
direction \cite{MNDS20}, noting some flaws.
In the second subsection, we present new theorems relating to
delineability, which partially rectify the flaws noted in the existing
work.
In the third subsection, we describe how the new theorems justify
the use of the fully reduced projection operation for the second
projection step under suitable assumptions.

\subsection{Review of existing work, with flaws described}
Section IV of \cite{MNDS20} (intended as a report on work in progress) 
presents a delineability criterion which has the aim of
justifying a reduction in the second projection operation
when constraints $f = 0$ and $e = 0$ are available at levels $n$ and
$n-1$, respectively. This criterion
uses a proposed notion of pointwise intersection order
of two coprime multivariate polynomials with complex number coefficients.
The notion proposed is recalled as follows.
When $f, g \in \C[x_1, \ldots, x_n]$ are coprime and $p \in \C^n$,
an affine change of coordinates is first made, if needed, to ensure
$p = 0$ and $f(0, \ldots, 0, x_n)$ is not identically zero;
hence $f(0, \ldots, 0, x_n)$ vanishes at $x_n = 0$ 
to some finite order $m \ge 0$.
Then Hensel's lemma or the Weierstrass preparation theorem is applied
to find a so-called {\em distinguished polynomial}
$h \in \C[[x_1, \ldots, x_{n-1}]][x_n]$ in $x_n$
(that is, $h$ is monic in $x_n$ and $h(0, \ldots, 0, x_n) = x_n^m$)
associated to $f$ near the origin.
Finally, the intersection order of $f$ and $g$ at $p$
is defined to be the order of the resultant $\mathrm{res}_{x_n}(h,g)$.
(This concludes the recollection of the notion proposed from
\cite{MNDS20}.)
However, our present author group, slightly different from that of 
\cite{MNDS20}, have subsequently noticed three issues with the work 
reported in that section of the article.

The first issue is that the definition recalled above of the
intersection order (or intersection multiplicity/number) 
of two coprime multivariate polynomials {\em at a point}
is nowhere to be found in the available standard
texts on algebraic geometry, such as \cite{CoxLittleOShea2005,
Harris1992, Hartshorne1977, Shafarevich2013}.
In fact, no other definition of such a pointwise intersection
order is found in such texts.
While this in itself does not invalidate the proposed notion,
it is consistent with the more serious problem explained below.
Of course, the case $n = 2$ is well understood and is discussed at 
length in some of the references mentioned and other works. 
But the case $n > 2$ 
is conspicuous by its absence. (However, the intersection number of $n$ 
polynomials in $n$ variables, having only finitely many common 
solutions, at a common solution point
is developed. Also, the notion of the intersection multiplicity of two 
hypersurfaces, having no common component, {\em along an irreducible
sub-variety of their intersection} is also defined.)

The second, more serious, issue is that the concept of the intersection
order of $f$ and $g$ at $p$ defined in \cite{MNDS20} is tacitly
assumed to be independent of the coordinate system in both the very definition of the concept
given and the proof of the delineability criterion (Theorem 7 of \cite{MNDS20}) which uses the concept. 
Now such coordinate independence is indeed true in the case $n = 2$.
However, it is not true in general. A counter-example is
furnished, with $n = 3$ and $(x_1, x_2, x_3) = (p, q, x)$, 
by setting $f(p, q, x) = x^3 + p x + q$
and $g(p, q, x) = \partial f/\partial x = 3 x^2 + p$.
Then $f$ is a distinguished polynomial in $x$, and we have
$\mathrm{res}_x(f,g) = 27q^2 + 4p^3$, which
has order 2 at the origin in the $(p,q)$-plane.
However, $g$ is a distinguished polynomial in $p$, and we have
$\mathrm{res}_p(f,g) = 2 x^3 - q$, which has order 1 at the origin
in the $(q,x)$-plane.
(It is intuitively clear that the intersection 
multiplicity of $f$ and $g$
along the curve $C$ of intersection $f = g = 0$ is 1,
since $\grad f = (x, 1, 3x^2 + p)$ and $\grad g = (1, 0, 6x)$
are linearly independent vectors at each point of $C$.
So, it seems that the former resultant, namely $\mathrm{res}_x(f,g)$,
misrepresents the intuitive idea of the intersection multiplicity of 
$f$ and $g$ at the origin.)

The third and relatively minor issue is that the first sentence of the
statement of the delineability criterion (Theorem 7 of \cite{MNDS20})
referred to above apparently has a typographical error: $e$ should be a polynomial in $x_1, \ldots, x_{n-1}$ (not in $x_1, \dots x_n$ as written in \cite{MNDS20}).

The next subsection aims to correct, as far as possible, these flaws.
The foundation is to replace the unsatisfactory notion of
pointwise intersection order introduced in \cite{MNDS20} by a somewhat different, simpler, coordinate-independent concept. 
The relevant theory is, of necessity,
completely redeveloped using the new concept.

\subsection{New theorems relating to delineability}
The theory in this subsection relies upon the notion of the order of a real analytic function $g(x_1, \ldots, x_n)$ relative to the real variety $V_f$ of another real analytic function $f(x_1, \ldots, x_n)$ at a non-singular point of $V_f$. We will now formulate this relative order notion in precise terms.

\begin{definition}
    Let $f$ and $g$ be real analytic functions defined in a neighbourhood of a point $p \in \R^n$. Let $V_f$ denote the real hypersurface defined by $f = 0$ near $p$, and suppose that $p$ is a non-singular point of $V_f$; that is, we suppose $f(p) = 0$, but some first-order partial derivative of $f$, say $\partial f/\partial x_i$ does not vanish at $p$. Then, by the implicit function theorem, $V_f$ is represented near $p$ by the graph of a real analytic function
    $x_i = \phi(x_1, \ldots, \hat{x}_i, \ldots, x_n)$ (where $\hat{x}_i$ denotes omission of the variable $x_i$). We say that $g$ has
    order $k \in \N \cup \{\infty\}$ relative to $V_f$ at $p$ if the order of the real analytic function 
\begin{equation}
\label{restriction_of_g_to_V}
    g(x_1, \ldots, x_{i-1}, \phi(x_1, \ldots, \hat{x}_i, \ldots, x_n),
        x_{i+1}, \ldots, x_n)
\end{equation}
    at $(p_1, \ldots, \hat{p}_i, \ldots, p_n)$ is $k$.
    We shall denote the order of $g$ relative to $V_f$ at $p$
    by $\mathrm{ord}_p g|_{f=0}$.
    In this notation, the expression $g|_{f=0}$ is a standard designation
    for the {\em restriction} of the function $g$ to the submanifold $V_f$, near $p$.  This restriction is expressed conveniently by means of local coordinates in Equation (\ref{restriction_of_g_to_V}).
\end{definition}

\begin{remark}
    The above definition is independent of the coordinates used to represent the submanifold $V_f$ near $p$. In particular, if also $\partial f/\partial x_j (p) \neq 0$, for $j \neq i$, and we use the real analytic function
    $x_j = \psi(x_1, \ldots, \hat{x}_j, \ldots, x_n)$ to represent
    $V_f$ near $p$, then we obtain the same value for
    $\mathrm{ord}_p g|_{f=0}$. Relevant details can be found in Appendix 2 of \cite{Whitney1972} and Chapter 2 of \cite{McCallum1985}.
\end{remark}

\begin{remark}
    A notion similar to, and more general than, the concept
    of relative order defined above was presented informally
    in a talk by C. W. Brown at SYNASC-2023
    (in the Special Session in Honour of Professor James Davenport).
    Brown's notion, called restricted order, is a key component
    of on-going collaborative development of a new projection
    operator for CAD.
\end{remark}

The following straightforward result will be  useful.

\begin{proposition} \label{Simple fact}
    With the notation and assumptions of the above definition,
 $\mathrm{ord}_p g|_{f=0} > 0$ if and only if $g(p) = 0$.   
\end{proposition}

\begin{proof} Observe that
$p_i = \phi(p_1, \ldots, \hat{p}_i, \ldots, p_n)$,
since $p \in V_f$ and $V_f$ is the graph of $\phi$, near $p$,
by the assumptions within the definition.
Suppose first that $g(p) = 0$.
Therefore
$g(p_1, \ldots, p_{i-1}, \phi(p_1, \ldots, \hat{p}_i, \ldots, p_n),
p_{i+1}, \ldots, p_n) = 0$ by the above observation.
Hence the order of the analytic function defined by Equation 
(\ref{restriction_of_g_to_V}) is positive or infinite,
from which $\mathrm{ord}_p g|_{f=0} > 0$ follows, by definition.

Conversely, suppose $\mathrm{ord}_p g|_{f=0} > 0$.
Then, the real analytic function defined by Equation (\ref{restriction_of_g_to_V}) vanishes at
$(p_1, \ldots, \hat{p}_i, \ldots, p_n)$, by definition.
But our initial observation says
$p_i = \phi(p_1, \ldots, \hat{p}_i, \ldots, p_n)$.
Therefore, $g(p) = 0$.
\end{proof}

We now proceed to state and prove a trilogy of 
theorems relating to delineability
which use the concept of relative order introduced above.
It is suggested that, for a first reading, the reader may like to
closely examine the statements of these theorems, and to leave the study
of the detailed proofs for another reading.
In reading the theorem statements, it may be helpful to keep in mind
the desired application to reduction of the second projection step. 
Our first significant result is a variation and consequence of the key lifting theorem (Theorem 2) from \cite{McCallum1998}.
Its wording is close to that of Theorem 7 of \cite{MNDS20},
keeping in mind that the key hypothesis that ``$d$ and $e$ are
intersection order invariant in $\sigma$'' in Theorem 7 of \cite{MNDS20}
is replaced by the assumption that ``the order of $d|_{e=0}$ is
invariant in $\sigma$'' in the theorem below.

\begin{theorem} \label{Variation of lifting theorem}
    Let $f(x_1, \ldots, x_n) \in \R[x_1, \ldots, x_n]$ be a real
    polynomial of positive degree in $x_n$, and let 
    $e(x_1, \ldots, x_{n-1}) \in \R[x_1, \ldots, x_{n-1}]$ be a real 
    polynomial of positive degree in $x_{n-1}$. 
    Let $d(x_1, \ldots, x_{n-1})$ and $a(x_1, \ldots, x_{n-1})$
    be the discriminant and the leading coefficient of $f$,
    respectively, both with respect to $x_n$, 
    suppose that $d$ is a non-zero polynomial, 
    and suppose that $e$ and $d$ are relatively prime.
    Let $\sigma$ be a connected submanifold of $\R^{n-1}$ contained in the
    real hypersurface (variety) defined by $e = 0$ in which 
    $a$ is sign-invariant and $f$ vanishes identically
    at no point. 
    Suppose that $\sigma$ contains no singular point of the real
    hypersurface defined by $e = 0$ and that 
    the order of $d|_{e = 0}$ is invariant
    in $\sigma$. 
    
    Then $f$ is analytic delineable on $\sigma$.
\end{theorem}

\begin{proof} 
In the case that $d$ vanishes nowhere in $\sigma$, the analytic 
delineability of $f$ on $\sigma$ follows immediately by
Theorem 3.1 of \cite{Brown2001} (which yields the degree-invariance of
$f$ on $\sigma$) and Theorem 2 of \cite{McCallum1998} (which delivers
the analytic delineability of $f$, using its degree-invariance 
and non-identically-vanishing on $\sigma$, 
and the order-invariance of $d$ in $\sigma$). 
So, henceforth assume that $d$ vanishes at some 
point of $\sigma$. Then, since the order of $d|_{e=0}$ is invariant in 
$\sigma$,
by assumption, $d$ must vanish everywhere in $\sigma$, 
by Proposition \ref{Simple fact}.
Since $\sigma$ is connected, it suffices to show that $f$ is analytic
delineable on $\sigma$ near an arbitrary point $p$ of $\sigma$.

Let $s$ be the dimension of $\sigma$. The case $s = n-1$, 
in which $\sigma$ is
just an open subset of $\R^{n-1}$, cannot arise by our assumption that
$d$ is non-zero, and our previous conclusion that $d$
vanishes everywhere in $\sigma$.
By our assumption that $e$ and $d$ are relatively prime, 
our previous conclusion that $d$ vanishes everywhere 
in $\sigma$, and the 
general form of Bezout's theorem (\cite{Abhyankar1990}, Lecture 30),
we may deduce that the order of $d|_{e=0}$ at $p$ is finite and 
positive, and that $s \neq n-2$. By excluding the trivial case $s = 0$, 
we therefore have $1 \le s \le n-3$.

Choose a coordinate system $\Phi~:~U \rightarrow U'$ about $p$,
with $U$ and $U'$ suitable neighbourhoods of $p$ and the origin,
respectively, with $\Phi$ informally denoted by $(y_1, \ldots, y_{n-1})$, 
such that
$\sigma$ is defined near $p$ by the equations 
$y_{s+1} = 0, \ldots, y_{n-1} = 0$ in the new coordinate system
(by Theorem 2.2 of \cite{McCallum1988}). We remark that the choice of 
such a coordinate system entails $\Phi(p) = 0$.

We denote by 
$f'(y_1, \ldots, y_{n-1}, x_n)$, $d'(y_1, \ldots, y_{n-1})$, 
$a'(y_1, \ldots, y_{n-1})$ and \\$e'(y_1, \ldots, y_{n-1})$
the polynomials $f$, $d$, $a$ and $e$ expressed in the new coordinates.
(We stress that the ``dash'' superscript does not denote derivative here.)
Then $d'$, $a'$ and $e'$ are real analytic functions defined near $0$
(to be precise, defined in $U'$),
and $f'$ is a polynomial in $x_n$ whose coefficients are real
analytic functions defined near $0$ (to be precise, defined in $U'$).
Let $T$ denote the $s$-dimensional linear
subspace of $\R^{n-1}$ defined by $y_{s+1} = 0, \ldots, y_{n-1} = 0$.
By the assumptions, $T$ is contained in the real hypersurface defined by 
$e' = 0$, near $0$, 
$a'$ is sign-invariant in $T$, near $0$,
and $f'$ vanishes identically at no point
of $T$, near $0$. Moreover, $0$ is not a singular point of
the real hypersurface $e' = 0$ and, by the proof of
Theorem 2.2.2 of \cite{McCallum1985} and a previous conclusion, 
the order of $d'|_{e'=0}$ is invariant, finite and positive in $T$, 
near $0$. 

It will be helpful to introduce simplified notation for the variables.
We shall put $x = (y_1, \ldots, y_s)$ and 
$y = (y_{s+1}, \ldots, y_{n-1})$.
Our next task is to show that we may deduce 
$\partial e'/\partial y_j (0) \neq 0$, for some $j$, 
with $s+1 \le j \le n-1$.
Suppose, for the sake of contradiction, that
$\partial e'/\partial y_j (0) = 0$, for every $j$, with 
$s+1 \le j \le n-1$. Then, since $0$ is a non-singular point of
$e' = 0$, we must have $\partial e'/\partial y_i (0) \neq 0$, for some
$i$, with $1 \le i \le s$. By continuity of $\partial e'/\partial y_i$,
we may refine the neighbourhood $U'$ of $0$ as needed to ensure that
$\partial e'/\partial y_i (q) \neq 0$, for all $q \in U'$.
Put 
\[T' = \{(x,y) \in U'~|~e'(x,y) = 0 \wedge 
                       y_{s+1} = 0 \wedge \cdots \wedge y_{n-1} = 0\}.\]
We claim that $T'$ is an $(s-1)$-dimensional submanifold of $\R^{n-1}$.
We prove our claim as follows. Take an arbitrary point $q \in T'$.
Since the $i$th component of $de'(q)$ is non-zero
(by the property of $U'$ ensured above), and
$dy_j(q)$ is the unit vector in the positive direction of the 
$y_j$-axis, for $s+1 \le j \le n-1$, it follows that the vectors
$de'(q), dy_{s+1}(q), \ldots, dy_{n-1}(q)$ are linearly independent.
Therefore, the set $T'$ defined above is an $(s-1)$-dimensional
submanifold of $\R^{n-1}$, by definition of submanifold
\cite{McCallum1988, Whitney1972}. Our claim is proved.
Yet our $s$-dimensional linear
subspace $T$ clearly satisfies $T \cap U' \subseteq T'$,
by definition of $T$ and the containment of $T \cap U'$ within 
the hypersurface $e' = 0$ noted above.
In other words, we have one submanifold of $\R^{n-1}$, namely
$T \cap U'$, contained in another, namely $T'$; yet the dimension of
the former submanifold is greater than that of the latter.
This contradicts a basic dimension property of submanifolds.
We may conclude that indeed
$\partial e'/\partial y_j (0) \neq 0$, for some $j$, 
with $s+1 \le j \le n-1$.

From here on, with slight abuse of notation, but hopefully no confusion,
we drop the ``dash'' decoration on $f$, $e$ and $d$.
We denote the $(n-s-2)$-tuple of variables
$(y_{s+1}, \ldots, \hat{y}_j, \ldots, y_{n-1})$ by $\eta$.
(Here $\hat{y}_j$ denotes omission of $y_j$.)
Recall that $e(x,y)$ is a real analytic function defined in $U'$,
$e(0,0) = 0$, $\partial e/\partial y_j (0) \neq 0$,
and the order of $d|_{e=0}$ is invariant, positive and finite
in $T \cap U'$. By the implicit function theorem,
there is an open box $B = \prod_{i=1}^{n-1} I_i$ 
centred at the origin in $\R^{n-1}$, and a function 
$\phi~:~(B^* = \prod_{i \neq j} I_i) \rightarrow I_j$ such that
$\phi$ is analytic, $B \subseteq U'$, and
for all $(x,y) \in B$,
we have $e(x,y) = 0$ if and only if $y_j = \phi(x, \eta)$.

For $(x, \eta) \in B^*$, put 
\[f^*(x, \eta, x_n) = f(x, y_{s+1}, \ldots, y_{j-1}, \phi(x, \eta),
                          y_{j+1}, \ldots, y_{n-1}, x_n).\]
Then $f^*(x, \eta, x_n)$ is a polynomial in $x_n$ whose coefficients
are real analytic functions defined in $B^*$.
Let us denote by $d^*(x, \eta)$ and $a^*(x, \eta)$ 
the discriminant and leading coefficient of $f^*(x, \eta, x_n)$,
respectively, both with respect to $x_n$. It is apparent that 
\[d^*(x, \eta) = d(x, y_{s+1}, \ldots, y_{j-1}, \phi(x, \eta),
                      y_{j+1}, \ldots, y_{n-1}),\]
for all $(x, \eta) \in B^*$, and a similar property holds for $a^*$. 
Let us set $T^*$ to be the $s$-dimensional
linear subspace of $\R^{n-2}$ (the space of the coordinates $(x, \eta)$)
defined by $\eta = 0$. It follows from the equation for $d^*$
immediately above,
together with the invariance and finiteness of the order of $d|_{e=0}$ in 
$T \cap U'$, that $d^*(x, \eta)$ is order-invariant, of finite order,
in $T^* \cap B^*$. Moreover, $a^*$ is sign-invariant in $T^* \cap B^*$,
and $f^*$ vanishes identically at no point of $T^* \cap B^*$.
Therefore, by Theorem 7.1 of \cite{McCallum2001},
which is a slight generalization of Theorem 2 of \cite{McCallum1998},
and an analogous generalization of Theorem 3.1 of \cite{Brown2001},
with the latter applied first to yield degree-invariance of $f^*$,
$f^*$ is analytic delineable on $T^* \cap B^*$.
It follows that $f$ is analytic delineable on $\sigma$ near $p$.
\end{proof}

Close inspection of the statement of Theorem \ref{Variation of lifting 
theorem} reveals that, to be useful for iterated equational 
projection, we need at least one more theorem which would yield as a 
conclusion the property related to a key hypothesis of
Theorem \ref{Variation of lifting theorem} about the invariance of the 
order of $d|_{e=0}$ in the submanifold $\sigma$.
In other words, we would like and aim to have a slight strengthening of
Theorem \ref{Main lemma for reduced equational projection},
which would yield as conclusion, in the notation of that theorem, 
invariance of
the order of $g$ relative to the real hypersurface $f = 0$ in each section
of $f$ over $S$ which contains no singular point of $f=0$. However,
despite some effort invested to find such a desired new theorem,
we have succeeded in proving only a rather special fact
in this direction (see Remark \ref{Special hypothesis} below).
In the statement of the next theorem, we have used notation similar to
that of Theorem \ref{Main lemma for reduced equational projection}.
For a second reading (of the following theorem statement), the reader is 
advised mentally to make the substitutions $n-1$ for $n$, $e$ for $f$, 
$d$ for $g$ and $\rho$ for $r$; this may help the reader to see the
intended application of the following theorem in conjunction with
Theorem \ref{Variation of lifting theorem}.

\begin{theorem} \label{Variation of main lemma from M 1999}
    Let $f$ and $g$ be real polynomials in $x_1, \ldots, x_n$,
    of positive degrees in $x_n$, let $r \in \R[x_1, \ldots, x_{n-1}]$
    be the resultant of $f$ and $g$ with respect to $x_n$,
    and suppose that $r$ is a non-zero polynomial.
    Let $S$ be a connected submanifold of $\R^{n-1}$ on which
    $f$ is analytic delineable and in which $r$ is order-invariant.
    Suppose moreover that the codimension of $S$ in $\R^{n-1}$
    is not more than 1.
    
    Then, for every section $\sigma$ of $f$ over $S$ which contains
    no singular point of the real hypersurface $f = 0$,
    the analytic function $g|_{f=0}$ is order-invariant in $\sigma$.
\end{theorem}

\begin{remark} \label{Special hypothesis}
    The special hypothesis, which we conjecture could be removed, is 
    that 
    concerning the codimension of $S$. Since we do not (yet) see how to 
    prove the theorem without this special hypothesis in force, we
    present the following proof which makes crucial use of this 
    hypothesis, amongst the other ones.
\end{remark}

\begin{proof}
    Let $\sigma$ be a section of $f$ over $S$ determined by the
    analytic function $\theta~:~S \rightarrow \R$.
    By connectedness of $\sigma$, it suffices to show that
    $g_{f=0}$ is order-invariant in $\sigma$ near an arbitrary point 
    $(p, \theta(p))$ of $\sigma$. Let $(p, \theta(p))$ be such a point.
    That $g_{f=0}$ is order-invariant in $\sigma$, near $(p, \theta(p))$
    follows easily by continuity of $g$ in the case that $g(p, \theta(p)) \neq 0$.
    So, henceforth we will assume that $g(p, \theta(p)) = 0$.
    Therefore, $r(p) = 0$, and hence $r$ vanishes everywhere in
    $S$, by the order-invariance assumption.
    Let $s$ be the dimension of $S$. The case $s = n-1$ (that is, the case in which the codimension of $S$ in $\R^{n-1}$ is 0) is one in which $S$ is an open subset of $\R^{n-1}$.
    This case cannot arise by our assumption that $r$ is a non-zero
    polynomial and our previous conclusion that $r$ vanishes everywhere
    in $S$. Therefore, by our codimension assumption,
    we have $s = n-2$ (that is, $S$ has codimension 1 in its space).

    By Theorem 2.2 of \cite{McCallum1988}, we choose coordinates
    designated informally by $(y,{y_n-1}) =
    (y_1, \ldots, y_{n-2}, y_{n-1})$ about $p$ such that $S$ is defined
    locally by the equation $y_{n-1} = 0$ in the new coordinate system,
    which maps $p$ to $0$, by definition.
    With slight abuse of notation, we denote by $f'(y, y_{n-1}, x_n)$,
    $g'(y, y_{n-1}, x_n)$ and $r'(y,y_{n-1})$ the polynomials
    $f$, $g$ and $r$ expressed in $(y,y_{n-1})$-coordinates.
    Then $r'$ is a real analytic function defined near $0$,
    and $f'$ and $g'$ are so-called pseudopolynomials in $x_n$ near $0$, 
    that is, polynomials in $x_n$ whose coefficients are real analytic.
    Moreover, by Theorem 2.1 of \cite{McCallum1988}, $r'$ is
    order-invariant in $T$ near $0$, where $T$ is the hyperplane in
    $\R^{n-1}$ defined by $y_{n-1} = 0$. For simplicity,
    we hereafter drop the ``dash'' decoration on $f$, $g$ and $r$,
    trusting that no confusion will arise.
    The analytic function $\theta$, for its part, is now similarly considered to be defined in $T$ (identified with $\R^{n-2}$), 
    near $0$, with graph 
    denoted by $\sigma$; and we may assume that
    $\theta(0) = 0$, without loss of generality.

    We claim that, in consequence of the assumption that $\sigma$
    contains no singular point of the real hypersurface $f=0$, near
    the origin, we may conclude that 
    $\partial f/\partial y_{n-1} (0) \neq 0$
    or $\partial f/\partial x_n (0)\neq 0$. 
    Suppose for the sake of contradiction that this is not the case.
    Then, by the non-singularity assumption, we must have
    $\partial f/\partial y_i (0) \neq 0$, for some $i$,
    with $1 \le i \le n-2$. Let $\sigma_i$ denote the intersection
    of the real hypersurface $f = 0$ with all the hyperplanes
    $y_1 = 0, \ldots, y_{i-1} = 0, y_{i+1} = 0, \ldots, y_{n-1} = 0$,
    in $\R^n$, near the origin. Then $\sigma_i \subseteq \sigma$ is an
    analytic 1-submanifold of $\R^n$, near the origin, which lies
    in the plane of the $y_i$ and $x_n$ coordinates.
    The tangent space (line) of $\sigma_i$ at $0$ is 
    the span of the vector
    $(0, \ldots, 0, 1)$, by Lemma 5C, Appendix II, of
    \cite{Whitney1972}. This contradicts the analyticity at $0$ of the
    function $\theta$ whose graph is $\sigma$. The claim is proved.

    We shall focus attention on the zero set of $f$ near the origin
    using an extension of Hensel's lemma (\cite{Abhyankar1990}, exercises
    on pages 92 and 95, with relaxation of monicity of $f$; or 
    \cite{McCallum1999}, Theorem 3.1, for
    the 3-variable case). Denote by $d$ the degree of $f$ with respect to
    $x_n$, and by $m$ the multiplicity of the root $0$ of $f(0,0,x_n)$.
    (Recall that, in the argument list for $f(0,0,x_n)$, the first $0$
    represents the first $n-2$ arguments, while the second represents
    $y_{n-1} = 0$.) By the extended Hensel's lemma just cited,
    there exists an open box $B \times (-\delta, \delta)$ 
    centred at the origin in $\R^{n-1}$
    and pseudopolynomials in $x_n$, $q^*(y, y_{n-1}, x_n)$ 
    of degree $d-m$ 
    and $h^*(y, y_{n-1}, x_n)$ of degree $m$, whose coefficients are 
    analytic in $B \times (-\delta, \delta)$, 
    such that $f = q^*h^*$, $h^*(0,0,x_n) = x_n^m$,
    and $q^*(0,0,0) \neq 0$. Since $q^*(0,0,0) \neq 0$ and $q^*$ is
    analytic, hence continuous, in $B \times (-\delta, \delta)$, 
    we may refine the box $B \times (-\delta, \delta)$
    and choose $\epsilon > 0$ such that $q^* \neq 0$ everywhere in
    $B \times (-\delta, \delta) \times (-\epsilon, \epsilon)$. 
    By continuity of $\theta$,
    after further refining $B$ if necessary, we may moreover assume that 
    for all $y \in B$, $\theta(y) \in (-\epsilon, \epsilon)$.
    It follows by delineability of $f$ and the above properties of $q^*$, 
    $h^*$ and $\theta$ that, for all $y \in B$,
    $\theta(y)$ is a root of $h^*(y, 0, x_n)$ of multiplicity $m$, 
    hence the {\em unique} root of $h^*(y, 0, x_n)$.
    By Theorem 3 of \cite{Loos1982} we have
    \[r = \mathrm{res}_{x_n}(q^*,g) \mathrm{res}_{x_n}(h^*,g).\]
    By hypothesis about the order-invariance of $r$ in $S$ 
    (hence in $T \cap (B \times (-\delta, \delta))$, in the new 
    coordinates), and Lemma A.3 of
    \cite{McCallum1988}, $v(y, y_{n-1}) = \mathrm{res}_{x_n}(h^*,g)$ 
    is order-invariant in $T \cap (B \times (-\delta, \delta))$. 
    Let us put 
    $\mu = \mathrm {ord}_{(0,0)} v$.

    To arrive finally at our desired conclusion 
    concerning the order of $g$
    relative to the real hypersurface $f = 0$ along $\sigma$
    we need to consider separately the two cases,
    at least one of which must occur as proved above, namely
    $\partial f/\partial x_n(0) \neq 0$ and 
    $\partial f/\partial y_{n-1} (0) \neq 0$.
    For the former case the remainder of the proof is relatively short.
    In this case, we have $m = 1$ and 
    $h^*(y, y_{n-1}, x_n) = x_n - \phi(y, y_{n-1})$,
    for some function $\phi$ analytic in $B \times (-\delta, \delta)$.
    Therefore, by Theorem 1 of \cite{Loos1982}, 
    $v(y, y_{n-1}) = g(y, y_{n-1}, \phi(y, y_{n-1}))$,
    for all $(y, y_{n-1}) \in B \times (-\delta, \delta)$.
    Hence, by definition, the order of $g|_{f=0}$ is invariant in
    $\sigma$, near the origin.

    On the other hand,  suppose that 
    $\partial f/\partial y_{n-1}(0)\neq 0$.
    We aim to prove that the order of $g|_{f=0}$ is invariant along 
    $\sigma$ near $0$, using this assumption.
    To prepare for this, it will be convenient to
    assume that the function $\theta$ is identically zero,
    which we may do without loss of generality by applying a suitable
    analytic change of coordinates ($y' = y$, $y_{n-1}' = y_{n-1}$,
    $x_n' = x_n - \theta(y)$).
    By the order-invariance of $v$ in 
    $T \cap (B \times (-\delta, \delta))$, and
    Lemma 4.4 of \cite{MPP2019} (using (3) implies (2)), 
    after refining $B \times (-\delta, \delta)$ about $0$
    if necessary, we may find 
    an analytic function $u$ in $B \times (-\delta, \delta)$, 
    with $u(y, y_{n-1}) \neq 0$
    for all $(y, y_{n-1}) \in B \times (-\delta, \delta)$, such that
    $v(y, y_{n-1}) = y_{n-1}^\mu u(y, y_{n-1})$,
    for all $(y, y_{n-1}) \in B \times (-\delta, \delta)$ 
    (where $\mu$ is the invariant order of $v$ in 
    $T \cap (B \times (-\delta, \delta))$, previously defined).
    Let $b$ be an arbitrary element of $B$.
    Then, by (a slight extension of) 
    the definition of the intersection multiplicity at the origin
    of two plane analytic curves in \cite{Abhyankar1990} (page 126), we may deduce that the intersection multiplicity at $(0,0)$ 
    of the plane analytic curves 
    $f(b, y_{n-1}, x_n) = 0$ and $g(b, y_{n-1}, x_n) = 0$ is $\mu$.
    We wish to use the invariance of the intersection multiplicity 
    at $(0,0)$ of these two plane analytic curves under 
    interchange of the
    coordinates, as discussed in \cite{Abhyankar1990}, Lecture 17.
    To this end, noting that neither plane analytic curve is necessarily
    polynomial in $y_{n-1}$, we apply the Weierstrass preparation
    theorem \cite{Abhyankar1990}, Lecture 16, to write
    \[f(y, y_{n-1}, x_n) = q'(y, y_{n-1}, x_n) h'(y, y_{n-1}, x_n),\]
    in $B \times (-\delta, \delta) \times (-\epsilon, \epsilon)$, and
    \begin{equation}
    \label{Weierstrass for g}
    g(b, y_{n-1}, x_n) = t'(y_{n-1}, x_n) x_n^\gamma k'(y_{n-1},x_n),
    \end{equation}
    in $(-\delta, \delta) \times (-\epsilon, \epsilon)$,
    where $q'$, $t'$ are everywhere non-vanishing 
    in their respective domains,
    $h'$, $k'$ are distinguished polynomials 
    in $y_{n-1}$, and $x_n^\gamma$ is the highest power of $x_n$ which
    divides $g(b, y_{n-1}, x_n)$.
    (For the application of Weierstrass to $f(y, y_{n-1}, x_n)$
    we need the multivariate version of the theorem described in 
    \cite{Abhyankar1990}, Lecture 16, Remark 5, page 120.)
    In fact, by our assumption that 
    $\partial f/\partial y_{n-1}(0)\neq 0$,
    we have $h'(y, y_{n-1}, x_n) = y_{n-1} - \psi(y, x_n)$,
    for some analytic function $\psi$ defined in 
    $B \times (-\epsilon, \epsilon)$.
    Therefore, by the invariance of the intersection multiplicity
    at $(0,0)$ of the two plane analytic curves 
    $f(b, y_{n-1}, x_n) = 0$ and $g(b, y_{n-1}, x_n) = 0$
    under interchange of the coordinates 
    \cite{Abhyankar1990} (Lecture 17, page 127), we have
    \[\mu = \mathrm{ord}_{x_n} \mathrm{res}_{y_{n-1}}
            (h'(b, y_{n-1}, x_n), x_n^\gamma k'(y_{n-1}, x_n)).\]
    By Equation (\ref{Weierstrass for g}), we have
    \[g(b, \psi(b, x_n), x_n) = 
      t'(\psi(b, x_n), x_n) x_n^\gamma k'(\psi(b, x_n), x_n),\]
    and by Theorem 1 of \cite{Loos1982},
    \[ \mathrm{res}_{y_{n-1}} (h'(b, y_{n-1}, x_n), 
                               x_n^\gamma k'(y_{n-1}, x_n)) = 
                               x_n^\gamma k'(\psi(b, x_n), x_n).\]
    Therefore, since $t' \neq 0$, 
    $\mathrm{ord}_{x_n} g(b, \psi(b, x_n), x_n) = \mu$.
    Application of Lemma 4.4 of \cite{MPP2019}(using (1) implies (3))
    yields the desired invariance of the order of
    $g(y, \psi(y, x_n), x_n)$ in $B \times (-\epsilon, \epsilon)$,
    from which the invariance of the order of $g|_{f=0}$ along $\sigma$,
    near the origin, follows by definition.
\end{proof}

We complete our trilogy of theorems relating to delineability with another
variation of Theorem \ref{Main lemma for reduced equational projection}.
The following variant is, like 
Theorem \ref{Variation of main lemma from M 1999}, also a slight 
strengthening of Theorem \ref{Main lemma for reduced equational 
projection}, achieved here
by weakening the hypothesis a little.

\begin{theorem}
    Let $f$ and $g$ be elements of $\R[x_1, \ldots, x_n]$,
    of positive degrees in $x_n$, let $r$ be the resultant of $f$ and $g$ 
    with respect to $x_n$, and suppose $r \neq 0$.
    Let $e$ be an element of $\R[x_1, \ldots, x_{n-1}]$,
    of positive degree in $x_{n-1}$, and suppose that $e$ and $r$
    are relatively prime. Let $\sigma \subseteq \R^{n-1}$ be a connected
    submanifold contained in the real hypersurface $e = 0$
    on which $f$ is analytic delineable.
    Suppose that $\sigma$ contains no singular point of $e = 0$
    and that the order of $r|_{e=0}$ is invariant in $\sigma$.
    
    Then $g$ is sign-invariant in each section of $f$ over $\sigma$.
\end{theorem}

\begin{proof}
    This proof will be modelled in part on that of 
    Theorem \ref{Variation of lifting theorem}, and in part on
    that of Theorem \ref{Variation of main lemma from M 1999}.
    Since many relevant details already appear 
    in the proofs of those theorems, this proof will be more concise.
    Let $\hat{\sigma}$ be a section of $f$ over $\sigma$ determined by
    the analytic function $\theta~:~\sigma \rightarrow \R$.
    By analogy with the proof of 
    Theorem \ref{Variation of main lemma from M 1999},
    it suffices to show that $g$ is sign-invariant in $\hat{\sigma}$ 
    near an arbitrary point $(p, \theta(p))$ of $\hat{\sigma}$,
    and we need to consider further only the case in which $r$
    vanishes everywhere in $\sigma$. Moreover, where $s$ denotes
    the dimension of $\sigma$, we may assume that $1 \le s \le n-2$.

    By Theorem 2.2 of \cite{McCallum1988},
    we choose an analytic coordinate system $\Phi$ about $p$
    in $\R^{n-1}$, informally denoted by $(y_1, \ldots, y_{n-1})$,
    such that $\sigma$ is defined locally by the equations
    $y_{s+1} = 0, \ldots, y_{n-1} = 0$ in the new coordinate system.
    Note that the choice of such a coordinate system entails 
    $\Phi(p) = 0$. With slight abuse of notation, we
    denote by $f(y_1, \ldots, y_{n-1}, x_n)$, 
    $g(y_1, \ldots, y_{n-1}, x_n)$,
    $r(y_1, \ldots, y_{n-1})$ and $e(y_1, \ldots, y_{n-1})$
    the original polynomials $f, g, r$ and $e$, respectively,
    expressed now in the new coordinates.
    We remind the reader that $f$ and $g$, so expressed,
    are pseudopolynomials in $x_n$, near $0$, and $r$ and $e$,
    so expressed, are analytic near $0$.
    Moreover, where $T$ is the linear subspace of $\R^{n-1}$
    defined by $y_{s+1} = 0, \ldots, y_{n-1} = 0$,
    the hypersurface $e = 0$ has no singular point in $T$, near $0$, and
    the order of the analytic function
    $r|_{e=0}$ is invariant in $T$, near the origin.
    The analytic function $\theta$, for its part, is now considered
    to be defined in $T$ (identified with $\R^s$), near $0$,
    with graph $\hat{\sigma}$, and we may assume that $\theta(0) = 0$,
    without loss of generality.

    By analogy with the proof of 
    Theorem \ref{Variation of lifting theorem}, we introduce
    simplified notation for the variables: namely, we put
    $x = (y_1, \ldots, y_s)$ and $y = (y_{s+1}, \ldots, y_{n-1})$.
    It follows analogously, 
    by the non-singularity assumption,
    that $\partial e/\partial y_j(0) \neq 0$, for some $j$, with
    $s+1 \le j \le n-1$. We denote the $(n-s-2)$-tuple of variables
    $(y_{s+1}, \ldots, \hat{y}_j, \ldots, y_{n-1})$ by $\eta$.
    (Here $\hat{y}_j$ denotes omission of $y_j$.)
    By the implicit function theorem, there is an open box
    $B = \prod_{i=1}^{n-1} I_i$ centred at the origin in $\R^{n-1}$
    and a function $\phi~:~(B^* = \prod_{i \neq j} I_i) \rightarrow I_j$
    such that $\phi$ is analytic, and for all $(x,y) \in B$,
    we have $e(x,y) = 0$ if and only if $y_j = \phi(x, \eta)$.

    For $(x, \eta) \in B^*$, put
    \[f^*(x, \eta, x_n) = f(x, y_{s+1}, \ldots, y_{j-1}, \phi(x, \eta),
                               y_{j+1}, \ldots, y_{n-1}, x_n),\]
    \[g^*(x, \eta, x_n) = g(x, y_{s+1}, \ldots, y_{j-1}, \phi(x, \eta),
                               y_{j+1}, \ldots, y_{n-1}, x_n).\]
    Then $f^*$ and $g^*$ are polynomials in $x_n$ whose coefficients
    are real analytic functions defined in $B^*$. Let us denote by
    $r^*(x, \eta)$ the resultant of $f^*(x, \eta, x_n)$ and
    $g^*(x, \eta, x_n)$ with respect to $x_n$. It is apparent that
    \[r^*(x, \eta) = r(x, y_{s+1}, \ldots, y_{j-1}, \phi(x, \eta),
                          y_{j+1}, \ldots, y_{n-1}),\]
    for all $(x, \eta) \in B^*$. Let us put $T^*$ equal to the
    $s$-dimensional linear subspace of $\R^{n-2}$ (the space of the
    coordinates $(x, \eta)$) defined by $\eta = 0$.
    It follows from the equation for $r^*$ immediately above,
    together with the invariance and finiteness of the order
    of $r|_{e=0}$ in $T$ near the origin, that $r^*$ is
    order-invariant, of finite order, in $T^* \cap B^*$
    (after suitable refinement of $B^*$, if necessary).
    Clearly, $f^*$ is analytic delineable on $T^* \cap B^*$,
    and $\hat{\sigma}$ (suitably restricted)
    may be considered to be a section of $f^*$ over
    $T^* \cap B^*$. Therefore, by Theorem 7.2 of
    \cite{McCallum2001}, $g^*$ is sign-invariant in $\hat{\sigma}$,
    near the origin. Hence $g$ is sign-invariant in $\hat{\sigma}$,
    near $(p, \theta(p))$.
\end{proof}

\subsection{Reducing the second projection step}
Let $n \ge 3$ and let $\phi^*$ be an input formula as
specified in Algorithm 1. Starting with $J_n = A$, the set of polynomials
occurring in $\phi^*$,
Step 3 of this algorithm repeatedly computes semi-restricted equational
projection sets $J_{j-1} = P_{E_j}^*(J_j)$; 
but it is remarked that, where possible,
fully restricted projection sets $P_{E_j}(J_j)$ should be used instead, 
for increased efficiency.
In fact, as recalled in Section 2.1, existing theory \cite{McCallum1999}
already permits the use of the fully restricted projection for the
first projection step; and an elementary argument 
shows that it can also be used for the ``last'' projection (of $J_2$,
see Proposition \ref{Bivariate case}).
The present subsection proves a result, combining the theorems of the
previous subsection, which states conditions under which the fully
restricted projection operation can also be used for the second projection
step.

We first recall some more of the notation of Algorithm 1, and make some
basic assumptions.
Suppose $f = 0$, with $f \in A$ assumed to be primitive and squarefree, 
is a level $n$ constraint of
the input formula $\phi^*$, with $d = \mathrm{discr}_{x_n}(f)$ 
and $a = \mathrm{ldcf}_{x_n}(f)$ non-zero polynomials.
We put $E_n = \{f\}$ and $J_{n-1} =P_{E_n}(J_n)$.
Then, by definition, $J_{n-1}$ is the union of $\mathrm{cont}(J_n)$
and $P_{F}(B)$, where $B$ is the finest squarefree basis for
$\mathrm{prim}(A)$ and $F$ is the set of irreducible factors of $f$.
Suppose that $e \in \Z[x_1, \ldots, x_{n-1}]$ is a product
of distinct irreducible projection factors at level $n-1$,
that $e = 0$ is a constraint of $\phi^*$,
and that $e$ and $a$ are coprime.
(Recall that $e=0$ may be either an explicit or implicit constraint.)
Put $E_{n-1} = \{e\}$, let $E$ denote the set of irreducible
factors of $e$, and set $J_{n-2} = P_{E_{n-1}}(J_{n-1})$.
We shall assume that $f$ and $e$ are {\em well-placed}: by this we mean
that either $\mathrm{discr}_{x_n}(f)$ is a product of elements of $E$,
or $e$ and $\mathrm{discr}_{x_n}(f)$ are coprime.

In order to state our next result we require one more concept:
we shall say that an element $g \in B \setminus F$ is
{\em well-positioned} with respect to $F$ and $E$ 
if either $\mathrm{res}_{x_n}(f,g)$ is a product of elements 
of $E$, or $e$ and $\mathrm{res}_{x_n}(f,g)$ are coprime.

\begin{theorem}
    Let $S \subseteq \R^{n-2}$ be a submanifold of codimension at most 1,
    suppose that each element of $J_{n-2}$ is order-invariant
    in $S$, and suppose that $S$ contains no point at which $e$
    vanishes identically. 
    
    Then $e$ is analytic delineable on $S$. 
    Moreover, $f$ is analytic delineable 
    on each section $\sigma$ of $e$ over $S$
    which contains no singular point of $e = 0$
    and no point at which $f$ vanishes identically,
    and each element of $B$ not in $F$ which is well-positioned with
    respect to $F$ and $E$
    is sign-invariant in every section of $f$ over each such $\sigma$.
\end{theorem}

\begin{proof}
    The fact that $e$ is analytic delineable on $S$ and order-invariant in
    each of its sections over $S$ follows from the assumptions
    by Theorem 3.1 of \cite{Brown2001} and Theorem 2 of 
    \cite{McCallum1998}. 
    The detailed reasoning is as follows.
    By an assumption, each element of $J_{n-2}$ is order-invariant 
    in $S$. Since each of $\delta = \mathrm{discr}_{x_{n-1}}(e)$ and
    $\alpha = \mathrm{ldcf}(e)$ is a product of certain elements
    of $J_{n-2}$,
    by definition of this projection set, $\delta$ is order-invariant
    in $S$ (by Lemma A.3 of \cite{McCallum1988})
    and $\alpha$ is sign-invariant in $S$.
    Another assumption states that 
    $e$ vanishes identically at no point of $S$.
    Therefore, by Theorem 3.1 of \cite{Brown2001}, $e$ is degree-invariant
    on $S$. Hence, by Theorem 2 of \cite{McCallum1998},
    $e$ is analytic delineable on $S$ and is order-invariant in
    each of its sections over $S$.
    
    Let $\sigma$ be a section of $e$ over $S$ which contains no singular
    point of the real hypersurface $e = 0$ and no point at which
    $f$ vanishes identically.
    We now show that $a$ is sign-invariant in $\sigma$.
    By assumption, $e$ and $a$ are coprime.
    By Theorem 3, in which we replace $n$ by $n-1$, and put
    $f = e$ and $g = a$ (hence $r = \mathrm{res}_{x_{n-1}}(e,a)$),
    the analytic function $a|_{e=0}$ is order invariant in $\sigma$,
    since $r \neq 0$ is order-invariant in $S$, by hypothesis
    (and Lemma A.3 of \cite{McCallum1988}, 
    if either $e$ or $a$ is reducible, or both are).
    Hence, $a$ is sign-invariant in $\sigma$, 
    by Proposition \ref{Simple fact} and the connectedness of $\sigma$.

    Next we show that $f$ is analytic delineable on $\sigma$.
    By assumption, $f$ and $e$ are well-placed.
    So, there are two cases to consider.
    In the first case, $d$ is a product of elements of $E$.
    Since $e$ is order-invariant in $\sigma$ (shown above),
    it follows by Lemma A.3 of \cite{McCallum1988} that
    $d$ is order-invariant in $\sigma$.
    Therefore, by Theorem 3.1 of \cite{Brown2001},
    $f$ is degree-invariant in $\sigma$.
    Hence, by Theorem 2 of \cite{McCallum1998},
    $f$ is analytic delineable on $\sigma$.
    In the second case, $e$ and $d$ are coprime.
    By Theorem 3, in which we replace $n$ by $n-1$,
    and put $f = e$ and $g = d$ (hence $r = \mathrm{res}_{x_{n-1}}(e,d)$),
    the analytic function $d|_{e=0}$ is order invariant in $\sigma$,
    since $r \neq 0$ is order-invariant in $S$, by hypothesis
    (and Lemma A.3 of \cite{McCallum1988}, if 
    either $e$ or $d$ is reducible, or both are).
    Therefore, by Theorem 2, $f$ is analytic delineable on $\sigma$.

    Now let $g \in B \setminus F$, with $g$ 
    assumed to be well-positioned with respect to $F$ and $E$. 
    There are two cases to consider.
    For the first case, we assume that $\mathrm{res}_{x_n}(f,g)$
    is a product of elements of $E$.
    Recall, as was shown above, that $e$ is order-invariant in $\sigma$.
    Therefore, by Lemma A.3 of \cite{McCallum1988},
    each element of $E$ is order-invariant in $\sigma$,
    and hence $\mathrm{res}_{x_n}(f,g)$ is order-invariant in $\sigma$.
    Therefore, by Theorem 1 (Section 2.1),
    $g$ is sign-invariant in every section of $f$ over $\sigma$.
    
    For the second case, we assume that $e$ and $\mathrm{res}_{x_n}(f,g)$
    are coprime.
    Let $f_i$ be an irreducible factor of $f$, of positive degree
    in $x_n$. Then $r_i := \mathrm{res}_{x_n}(f_i,g)$ and $e$ are coprime,
    since $r_i$ is a factor of $r$,
    and $r_i \in J_{n-1}$, by definition.
    By Theorem 3, in which we replace $n$ by $n-1$,
    and put $f = e$ and $g = r_i$ 
    (hence $r = \mathrm{res}_{x_{n-1}}(e,r_i)$),
    the analytic function $r_i|_{e=0}$ is order invariant in $\sigma$,
    since $r \neq 0$ is order-invariant in $S$, by hypothesis
    (and Lemma A.3 of \cite{McCallum1988}, if $e$ or $r_i$
    is reducible, or both are).
    Finally, by Theorem 4, $g$ is sign-invariant in every section
    of $f_i$ over $\sigma$, from which it follows that 
    $g$ is sign-invariant in every section of $f$ over $\sigma$.
\end{proof}

    We conclude this subsection with a few remarks. First, let us
    carefully consider the two assumptions of Theorem 5 that 
    ``$S$ contains no point at which $e$ vanishes identically'' and
    ``$\sigma$ ... contains ... no point at which 
    $f$ vanishes identically''. The reader may wonder
    if these are appropriate, since of the coefficients of
    $f$, only the leading one is placed into the first projection set 
    $J_{n-1}$, and likewise for $e$ in relation to $J_{n-2}$. 
    So, why is it reasonable to expect
    that these assumptions will hold in practical applications of
    Theorem 5? The reason is that our use of the Brown-McCallum
    projection operator entails the analysis and solution 
    of projection factor coefficient systems prior to the lifting phase
    of CAD. The solution of such systems is required only when it is
    determined that there are finitely many such solutions,
    as the projection operator fails otherwise.
    Moreover, such solutions -- the isolated points at which projection
    factors vanish identically -- must be ``added'' to the CAD
    progressively during the stack construction phase \cite{Brown2001}.
    So, in practical applications of Theorem 5, the submanifolds
    $S$ and $\sigma$ (which represent arbitrary cells of the CADs
    of $\R^{n-2}$ and $\R^{n-1}$, respectively) will satisfy
    the assumptions quoted above.

    Second, the assumptions of Theorem 5 that 
    ``$S$ has codimension at most 1'' and ``$\sigma$ ... contains
    no singular point of $e = 0$'' are admittedly restrictive.
    So, we offer here only a partial validation of the use
    of the fully reduced projection for the first two steps,
    given the availability of the equational constraints $f = 0$
    and $e = 0$ stipulated. This is enough to justify the use
    of the fully reduced projection for all $n \le 4$,
    provided that the codimension and non-singularity conditions hold
    in the case $n = 4$.
    Also, as we show in the next section,
    there are specific applications where $n > 4$
    which could benefit from
    the application of the theory we have presented in this paper.
    Moreover, we hope and trust that our results will stimulate
    further progress in this direction.

    Third, we offer a simple way to view the results we have presented.
    In rough terms, we have shown that, under the conditions
    specified and assuming that $d$ and $e$ are coprime, 
    we can omit the ``troublesome''
    $\mathrm{discr}_{x_{n-1}}(d)$ from the second projection set 
    $J_{n-2}$. We can also omit $\mathrm{discr}_{x_{n-1}}(p)$
    from $J_{n-2}$, for each (other) nonconstraint
    projection polynomial $p$ at level $n-1$.
    This is a hopeful result, since the inclusion
    of these polynomials, and others spawned by them during
    iterated projection, is a significant component of the
    ``doubly-exponential wall'' of CAD-based QE in the presence of
    multiple ECs \cite{DEMU25}.

\newpage 

\section{Detailed discussion of examples}

In this section we revisit two of the examples introduced in Section 1,
in the light of the theory presented in Sections 3 and 4.

The experiments with program QEPCADB which we report were conducted on
a Unix server with an Intel Xeon Gold CPU (2.30 GHz). For each
experiment, 1 megabyte of memory was made available.

\subsection{The example of Fortuna, Gianni and Trager}
We present details about a system of generic solutions
to Example 2 from Section 1, introduced by \cite{FGT01},
which we found with the help of algebra and logic software.

Algorithms 1 and 2 (Section 3) have not yet been implemented.
However, the existing program QEPCADB can be used to simulate
part of the work of Algorithm 1.
QEPCADB allows the user to select the projection operation to be used at each level from a short list of available operators.
For our first experiment, we selected Hong's projection to be used at
all levels, since this is the most efficient operator with no exceptions for its use which is currently available.
We used the quantified formula $\phi^*$ defined in Example 2 (Section 1)
as input, with the variable ordering $(c,b,a,z,y,x)$:
\[\phi^* \equiv (\exists z)(\exists y)(\exists x) \phi(c,b,a,z,y,x).\]
The program delivered the expected solution \texttt{TRUE} 
after 14 milliseconds.
This solution was found by constructing a suitable partial CAD of
$\R^6$, considered to be part of a full CAD $\mathcal{D}$
of $\R^6$, in theory.
The CAD $\mathcal{D}_3$ of $(c,b,a)$-space (that is, the parameter space)
induced by $\mathcal{D}$
contains 63 cells. The first cell listed has index $(1,1,1)$
and sample point $(-2, -3, -1)$. All the 6 projection factors at levels
no greater than 3 have negative sign on this cell.
Similar information about each of the other 62 cells could, 
as for the first cell,
be displayed using the command \texttt{d-cell}.
The program issued a warning about the possible use
of McCallum's projection (though this was not selected), since
the set $A$ of polynomials occurring in the input is not well-oriented.
In fact, it is possible to use McCallum's projection provided that
an assumption is used to rule out occurrence of nullifying cells
of positive dimension for the first program run. 
(If this is to be done, then further runs of the program are needed 
to treat the special cases left out by the first.)
Apart from small reductions in computation time and cell count, 
there was not much
benefit to be gained by use of McCallum's projection for this example
at this stage.

A succinct and insightful enhanced solution, in the spirit of
Sections 1 and 3 above, was found using key information provided by
QEPCADB (and alternatively by Gr\"{o}bner basis computation).
The authors of \cite{FGT01} refer to a Gr\"{o}bner basis
for the ideal generated by the polynomials in $A$.
This basis contains the polynomial
\[\rho(c,b,a,z) = z(acz + az + b^2 - bc - b),\]
which is precisely the level 4 polynomial identified by
QEPCADB as the unique implicit constraint of $\phi^*$ at this level.
Inspection of $\rho$ suggests that it may be helpful first to consider
carefully the assumption $a(c+1) \neq 0$.
The reason is that, clearly, for every fixed $(c,b,a)$ satisfying
this assumption, there are at most two solutions of
the equation $\rho(c,b,a,z) = 0$, hence at most two solutions of
the system $\phi(c,b,a,z,y,x)$.
Indeed, the solutions of the latter system may be found
by explicitly solving $\rho(c,b,a,z) = 0$,
obtaining $z = 0$ and $z = b(1-b+c)/(ac+a)$,
then back substituting these solutions into the system $\phi$.
We find $(x,y,z) = (0,0,0)$ and
\[(x,y,z) = \left(\frac{b-c-1}{a(c+1)}, \frac{b(1-b+c)}{a(c+1)^2},
            \frac{b(1-b+c)}{a(c+1)}\right).\]
Moreover, close inspection of these solution forms suggests that it
may be worthwhile to slightly strengthen the assumption:
$a(b-c-1)(c+1) \neq 0$. The reason is that, for each fixed $(c,b,a)$
which satisfies the stronger assumption, there are {\em exactly}
two distinct solutions $(x,y,z)$.
We see that a system of generic solutions $\mathcal{O}$ for $\phi^*$
is the collection
of 8 connected open slightly skewed ``octants'' determined (bordered)
by the planes $a=0$, $b-c-1 = 0$ and $c+1 = 0$ in $\R^3$.
In each octant, the number $\nu$ of distinct solutions $(x,y,z)$
of $\phi(c,b,a,z,y,x)$ is 2,
and $(x,y,z)$ is expressed in terms of $(a,b,c)$ by the above formulas.
Consideration of each bordering plane in turn reveals that,
at every point of the variety defined by $a(b-c-1)(c+1) = 0$,
we have $\nu = 1$ or $\nu = \infty$.
In light of this, we may conclude that the system of generic
solutions $\mathcal{O}$ described above is minimal (in the sense that the
cardinality of $\mathcal{O}$ is minimum).

We ran program QEPCADB with input $\phi^*$, issuing the command
\begin{verbatim}
    assume [a(b - c - 1)(c + 1) /= 0].
\end{verbatim}
and selecting McCallum's projection (which is valid under
this assumption). The program found the expected solution formula
\texttt{TRUE} in 7 milliseconds.
The induced CAD of $(c,b,a)$-space has 24 cells labelled \texttt{T}
(and 12 cells not labelled, because they do not satisfy the assumption).
Presumably, the minimal system described above could in principle
be obtained by pasting together adjacent cells for which $\nu = 2$.
While the command \texttt{d-cell I} gives some information
about the cell whose index is \texttt{I}, it does not report
the number of associated solutions 
of the system $\phi$ nor how the solutions
depend on the parameters. It would seem to be reasonably straightforward
to enhance the capability of QEPCADB to provide such information,
as outlined in Section 3.

\subsection{More about Solotareff's problem}
Recall that Example 3 (Section 1) discussed Solotareff's approximation 
problem: find the best approximation to a real polynomial of
degree $n$ by a real polynomial of degree at most $n-2$.
(As in Example 3, $n$ does not denote the number of variables,
rather the degree of the given polynomial, for consistency
with the existing literature.) Example 3 reviewed a computer-assisted
solution to the problem in the case $n = 3$.
This subsection carefully considers the case $n = 4$,
for which detailed treatments  -- excepting the complete mathematical
justification of the equational projection method used --
are found in \cite{Collins1996, Collins1998}.
Recall that in this case the problem is to find the
best approximation to $x^4 + rx^3$ by a polynomial of the
form $ax^2 + bx + c$; and we need to consider only the case $r > S_4$,
since otherwise theory delivers an explicit solution.

In this case we have $S_4 = 12 - 8 \sqrt{2}$,
the smaller root of $x^2 - 24x + 16$.
Consequently, we can express $r > S_4$ by the formula
$r^2 - 24r + 16 < 0 \vee [r > 1 \wedge r^2 -24r + 16 \ge 0]$.
We also note that $c = 1 - a$, from theory.
Thus, for $r > S_4$, we can formulate the problem of expressing
$a$ in terms of $r$ as the following QE problem (\cite{Collins1998},
Section 8). However the quantified formula which appears in 
\cite{Collins1998}, Section 8, has some typographical errors 
which are corrected
in \cite{McCallum1998}, Section 7, and as follows in our formula $\phi^*$:
\begin{align}
(\exists b)(\exists u)(\exists v) &[[r^2 -24r + 16 < 0
                        \vee [r > 1 \wedge r^2 -24r + 16 \ge 0]]
 \nonumber \\
&\wedge r - b > 0 \wedge -1 < u \wedge u < v \wedge v < 1 
 \nonumber \\
&\wedge v^4 + rv^3 - av^2 - bv - (1-a) = r - b \label{f1}
\\ 
&\wedge u^4 + ru^3 - au^2 - bu - (1-a) = r - b \label{f2}
\\
&\wedge 4v^3 + 3rv^2 - 2av - b = 0 \label{f1'}
\\
&\wedge 4u^3 + 3ru^2 - 2au - b = 0]. \label{f2'}
\end{align}
\cite{Collins1996} sets up and solves an equivalent, 
slight variation of the above.

We aim here to justify the use of the fully reduced equational projections
in the solutions provided by \cite{Collins1996, Collins1998}.
The variable ordering chosen was $(r,a,b,u,v)$.
Inspecting the above formulation, we see that there
are four explicit equational constraints, two in the variable $v$
and two in $u$. Let us denote by $f_1(r,a,b,u,v)$ the polynomial
obtained by shifting all terms to the left hand side of the equation
on line (\ref{f1}) in the above formula; notice that the
equation on line (\ref{f1'}) is then $f_1' = 0$, where $f_1'$ denotes
the derivative of $f_1$ with respect to $v$.
Thus, the two constraints in $v$ are $f_1 = 0$ and $f_1' = 0$.
Similarly, denote by $f_2(r,a,b,u)$ the polynomial obtained
by shifting all terms to the left hand side of the equation
on line (\ref{f2}) above; then the equation on line (\ref{f2'}) is $f_2' = 0$,
where this time $f_2'$ denotes the derivative with respect to $u$.
The two constraints in $u$ are $f_2 = 0$ and $f_2' = 0$.

For the first projection, $f_1$ is chosen as the pivot and 
hence $\mathrm{discr}_v(f_1) = \mathrm{res}_v(f_1, f_1')$
occurs ``twice'' (in computation only once, of course)
in the first projection set $J_4$, amongst other polynomials.
Let $e_1(r,a,b)$ denote the greatest squarefree divisor of
$\mathrm{discr}_v(f_1)$. Then $e_1$ is the product of two 
irreducible factors, say $p_1$ and $q_1$;
and $e_1 = 0$ is an {\em implicit} constraint at level 3.
For the second projection, $f_2$ is chosen as the pivot;
hence $\mathrm{discr}_u(f_2) = \mathrm{res}_u(f_2, f_2')$
occurs in the second projection set $J_3$, amongst other polynomials.
Let $e_2(r,a,b)$ denote the greatest squarefree divisor of
$\mathrm{discr}_u(f_2)$. Then $e_2$ is the product of two
irreducible factors, say $p_2$ and $q_2$;
and $e_2 = 0$ is another implicit constraint at level 3.
For the third projection, $e_2$ is selected as the pivot constraint.
The third projection set $J_2$ hence contains (factors of)
$\delta(r,a) := \mathrm{discr}_b(e_2)$ and 
$\rho(r,a) := \mathrm{res}_b(e_2, e_1)$, amongst other polynomials.
Collins observed that the latter polynomial $\rho(r,a)$ -- the propagated
bivariate constraint -- has seven irreducible factors;
and he identified one of these factors,
which we shall denote by $\rho^*(r,a)$, as crucial to his solution:
\begin{align*}
&324a^4 + 324r^2 a^3 - 2016a^3 + 108r^4 a^2 - 1128r^2 a^2 + 4576a^2
\\
&\qquad +12r^6 a - 224r^4 a + 1392r^2 a - 4480a - 15r^6 + 112r^4 - 608r^2
                                                           + 1600.
\end{align*}
We shall state the solution shortly; but for now we continue
our step-by-step summary of its working out.
For the fourth and last projection, $\rho(r,a)$ is used as the pivot,
with the univariate projection set 
$J_1 = P_{\{\rho\}}(J_2)$ obtained accordingly.
QEPCAD then obtains an irreducible basis for $J_1$,
and constructs a CAD $\mathcal{D}_1$ of $\R^1$ 
for which each element of this basis,
and hence of $J_1$, is order-invariant in every cell of $\mathcal{D}_1$.
Proposition \ref{Bivariate case} may now be applied,
showing that $\rho(r,a)$ is analytic delineable on each 1-cell $c$
of $\mathcal{D}_1$, and that each element of $J_2$ is order-invariant
in every such section. 
(Strictly speaking, Proposition \ref{Bivariate case} makes equivalent conclusions about irreducible bases for the sets $\{\rho\}$
and $J_2$, respectively.)
Finally, Collins used the interactive stepwise stack construction
capability of QEPCAD to construct a claimed partial CAD $\mathcal{D}$
of $\R^5$ sufficient to infer the solution involving $\rho^*$
to the given QE problem. Careful reading of Collins'
work reveals that his solution actually comprises two separate, related
parts: a theorem $T$ about $\rho^*$ which he discovered, and
a formula $\phi'$ equivalent to the given $\phi^*$, whose definition
relies upon the theorem $T$.
The theorem $T$ states that for every $r > S_4$, $\rho^*(r,a)$
(as a polynomial in $a$) has exactly two distinct real roots.
The solution formula $\phi'$ may then be expressed as follows:
``$r > S_4$ and $a$ is the larger real root of $\rho^*(r,a)$''.
An astute reader will notice that the formal expressions of both
theorem $T$ and formula $\phi'$ in the language of Tarski algebra may 
require quantifers, so $T$ and $\phi'$ are not, 
strictly speaking, quantifier-free.

To fulfil our aim of justifying Collins' solution and methodology
we shall use Theorem 5 to verify the claimed partial CAD
$\mathcal{D}$.
By the known existence and uniqueness of a solution to
Solotareff's problem \cite{Achieser1956},
it will suffice to verify that, for every 1-dimensional section
$S \subset \R^2$ of $\rho^*$,
for which $r > S_4$ and $a$ is the larger real root of $\rho^*(r,a)$,
the pivot constraints used at levels 3, 4 and 5 determine a tree
$\mathcal{T}_S$ of cells such that the root node is the cell $S$,
and the children of each node $n$ are the cells in a stack over $n$.
Additionally, $\mathcal{T}_S$ should have the properties that
every leaf is at the same level, and the quantifier-free
part $\phi$ of the given $\phi^*$ is true throughout at least one
leaf of $\mathcal{D}_S$.

Let $S \subset \R^2$ be an arbitrary such 1-dimensional 
section of $\rho^*$. By Proposition \ref{Bivariate case},
as noted above, each element of $J_2$ is order-invariant in $S$.
Firstly, we shall use Theorem 5 to verify the delineability of the
pivot constraint polynomials $e_2$ and $f_2$ used at levels 3
and 4, respectively. We will see that $e_2$ is delineable over $S$,
and $f_2$ is delineable over every section of $e_2$ over $S$.
Indeed, in the notation of Theorem 5,
let us put $f = f_2$, $d = \mathrm{discr}_u(f_2)$,
$a = \mathrm{ldcf}_u(f_2) = 1$, and $e = e_2~(= d)$.
Noticing that $f$ and $e$ are well-placed (because $d = e$),
$e$ and $a$ are coprime (since $a = 1$), and $S$ contains no
point at which $e$ vanishes identically (since $\mathrm{ldcf}_b(e) = 27$),
we see that all the hypotheses of Theorem 5 hold.
Therefore, by Theorem 5, $e$ is analytic delineable on $S$;
hence the sections of $e$ over $S$ determine a stack over $S$.
Let $\sigma$ be a section of $e$ over $S$.
By inspection of the detailed output of QEPCAD,
we see that $\rho^*$ and $\mathrm{discr}_b(e)$ are coprime.
Hence, by the order-invariance of $\mathrm{discr}_b(e)$
(a product of elements of $J_2$) in $S$, 
we must have $\mathrm{discr}_b(e) \neq 0$
throughout $S$. It follows that $\sigma$ contains no singular
point of the real algebraic surface $e = 0$.
Again, by Theorem 5, $f$ is analytic delineable on $\sigma$;
hence the sections of $f$ over $\sigma$ determine a stack over $\sigma$.
With $B$ denoting the set of level 4 irreducible projection factors
and $F$ the set of irreducible factors of $f$,
we may observe that every $g \in B \setminus F$ is well-positioned
with respect to $F$ and $E$. Once again, by Theorem 5,
each such $g$ is sign-invariant in every section of $f$ over $\sigma$.

Secondly, we shall employ a different application of Theorem 5
to verify the delineability of the level 5 constraint polynomial $f_1$ -- as a polynomial in the 4 variables $r, a, b, v$ -- 
over the arbitrary section $\sigma \subset \R^3$ of $e_2$ over $S$ 
introduced in the previous paragraph.
Indeed, in the notation of Theorem 5, we now put $f = f_1$,
$d = \mathrm{discr}_v(f_1)$, $a = \mathrm{ldcf}_v(f_1) = 1$
and $e = e_2$. Noticing that $f$ and $e$ are well-placed
(because now $e$ and $d$ are coprime),
and $e$ and $a$ are coprime (since $a = 1$),
we see again that all the hypotheses of Theorem 5 hold.
Therefore, by Theorem 5, $f$ is analytic delineable on $\sigma$,
and consequently $f$ is analytic delineable 
on every section $\hat{\sigma}$ of $f_2$ over $\sigma$.
Hence the sections of $f$ over $\hat{\sigma}$ determine
a stack over $\hat{\sigma}$.
With $B$ now denoting the set of all level 5 irreducible projection factors and $F$ now denoting the set of irreducible factors of $f$,
we may observe that every $g \in B \setminus F$ whose variables
lie in $\{r, a, b, v\}$ is well-positioned with respect to $F$ and $E$.
By Theorem 5, each such $g$ is sign-invariant in every section
of $f$ over $\hat{\sigma}$.
Since $v - u$, the only exception, is linear in $v$,
it is not difficult to see that it too is sign-invariant in
every section of $f$ over $\hat{\sigma}$.
This concludes our validation of the existence of the
claimed tree $\mathcal{T}_S$.

We make two remarks in conclusion of this subsection.

\begin{remark}
    Collins' solution method for this problem instance is 
    essentially a variant of Algorithm 1 (Section 3).
    For this variant, the number $p$ of parameters is 1,
    and the number $u$ of unknowns is also 1: indeed, the parameter
    is $r$ and the unknown is $a$. A slight change in terminology
    is required: for $\alpha \in \R$, the value $\beta \in \R$ 
    is a solution associated with $\alpha$ if $\phi^*(r,a)$ is true.
    In this variant of Algorithm 1,
    using the notation of that algorithm, pivots and projection sets
    are computed down to $e_2$ and $J_1$, respectively, in Steps 1, 2,
    and 3. In Step 5, we compute the open intervals of a $J_1$-invariant
    CAD $\mathcal{D}_1$ of parameter space $\R^1$, 
    together with solution cardinality
    information. From theory, $\nu = 1$ holds for
    each open interval $c$ with $r > S_4$ in this problem
    instance. Finally, in Step 6, for each open interval $c$ with
    $r > S_4$, we construct the defining formula for $a$ 
    in terms of $r \in c$.
    Perhaps most important of all to note is that Collins
    \cite{Collins1996, Collins1998} used the fully reduced projection
    operation $P_{E_j}(J_j)$ at each projection step,
    although this has been rigorously validated 
    in general only for the cases
    $j = n$ \cite{McCallum1999} and $j = 2$ 
    (Proposition \ref{Bivariate case} of the present paper).
    Yet, we have outlined above a new justification 
    for the use of the fully reduced projection operation
    for $2 < j < n$ also, for this problem instance.
\end{remark}

\begin{remark}
    Daniel Lazard \cite{Lazard2006} described an adaptation of the
    concepts and methods of \cite{LR07} which is effective in
    solving Solotareff's problem up to $n = 9$ in a certain sense.
    No explicit solution formulas for the desired coefficients of the
    best approximation polynomial in terms of the parameter $r$
    were provided, though these could presumably be obtained in
    principle. In any case, our aim here in this paper is to enhance
    CAD-based QE as a general purpose tool when multiple equations are 
    present. We simply illustrated the potential benefit of our
    enhancements using Solotareff's problem, rather than making
    any claim about superior performance relative to other methods
    for this particular problem.
\end{remark}

\section{Conclusion and planned further work}
We summarize the work reported in this paper. Sections 3 and 4 are 
the technical heart of the new contributions that we present.
Section 3 contains a formal algorithm which, given a parametric
system in the form of an existentially quantified formula
$\phi^* \equiv (\exists x_{n-k+1}) \cdots (\exists x_n) 
               \phi(x_1, \ldots, x_n)$, where $s = n-k$
is the number of parameters, finds a description of
a system of generic solutions for $\phi^*$.
The details concerning the structure of the expected input $\phi^*$
and the constitution of the generic solution description in the output
are provided in the specification of Algorithm 1.
A supporting proposition and worked examples are included.
Elaboration of the general definition of the candidate constraint set
at each level is needed.

Section 4 introduces the concept of the order of a real analytic function
$g(x_1, \ldots, x_n)$ relative to the real variety $V_f$
of another real analytic function
$f(x_1, \ldots, x_n)$ at a non-singular point $p$ of $V_f$: 
in symbolic notation, this is denoted by
$\mathrm{ord}_p g|_{f=0}$.
A trilogy of new delineability related theorems using this concept
is presented. The purpose of this theory is to validate the use
of the fully reduced equational projection operation for the
second projection step of Algorithm 1 and its variants 
under suitable assumptions.
This is a relatively small, but useful, step toward ``pushing back
the doubly exponential wall'' of CAD 
when many equations are present in $\phi$.

It would be desirable to improve Theorem 3 by removing the hypothesis
about the codimension of $S$. Further to that,
a simplified and extended version of the theory presented in Section 4
would be a worthwhile aim in the future.
Perhaps this could be achieved using the concept of the multiplicity
of a variety at a point or a closely related notion.
It would be useful to implement the algorithms of Section 3
and evaluate them experimentally.
Investigation of the idea of replacing 
the single projection step paradigm of CAD
using a ``multi-step'' projection operation,
when many equational constraints are available,
would also seem to be a promising direction.

\bibliography{references.bib}

\end{document}